\documentclass{article}

\usepackage{arxiv}

\usepackage[utf8]{inputenc} 
\usepackage[T1]{fontenc}    
\usepackage{hyperref}       
\usepackage{url}            
\usepackage{booktabs}       
\usepackage{amsfonts}       
\usepackage{nicefrac}       
\usepackage{microtype}      
\usepackage{lipsum}		
\usepackage{graphicx}
\usepackage{natbib}
\usepackage{doi}
\usepackage{amsmath}
\usepackage{amsthm}
\usepackage[english]{babel}
\usepackage{appendix}
\usepackage{bbm}
\usepackage{cleveref}
\usepackage{subcaption}
\usepackage{multirow}
\usepackage{multicol}
\usepackage{tabularx}

\newtheorem{assumption}{Assumption}
\newtheorem{remark}{Remark}
\newtheorem{theorem}{Theorem}
\newtheorem{corollary}{Corollary}
\newtheorem{lemma}{Lemma}
\newtheorem{proposition}{Proposition}

\DeclareMathOperator{\Var}{\text{Var}}
\DeclareMathOperator{\Cov}{\text{Cov}}

\title{Causal inference in social platforms under approximate interference networks}

\author{{\hspace{1mm}Yiming Jiang} \\
	Industrial and System Engineering\\
	Georgia Institute of Technology\\
	Atlanta, GA 30318 \\
	\texttt{yjiang463@gatech.edu} \\
	\And
	{\hspace{1mm}Lu Deng} \\
	Tencent, Inc\\
         Shenzhen, Guangdong, China\\
	\texttt{adamdeng@tencent.com} \\
 	\And
	{\hspace{1mm}Yong Wang} \\
	Tencent, Inc\\
        Shenzhen, Guangdong, China\\
	\texttt{darwinwang@tencent.com} \\
  	\And
	{\hspace{1mm}He Wang} \\
		Industrial and System Engineering\\
	Georgia Institute of Technology\\
	Atlanta, GA 30318 \\
	\texttt{he.wang@isye.gatech.edu} \\
}

\renewcommand{\shorttitle}{Causal Inference Under Surrogate Networks}

\hypersetup{
pdftitle={Causal Inference Under Surrogate Networks},
pdfauthor={Yiming Jiang, Lu Deng, Yong Wang, He Wang},
pdfkeywords={Causal inference, Network interference, SUTVA},
}

\begin{document}
\maketitle

\begin{abstract}
Estimating the total treatment effect (TTE) of a new feature in social platforms is crucial for understanding its impact on user behavior. However, the presence of network interference, which arises from user interactions, often complicates this estimation process. Experimenters typically face challenges in fully capturing the intricate structure of this interference, leading to less reliable estimates. To address this issue, we propose a novel approach that leverages surrogate networks and the pseudo inverse estimator. Our contributions can be summarized as follows: (1) We introduce the surrogate network framework, which simulates the practical situation where experimenters build an approximation of the true interference network using observable data. (2) We investigate the performance of the pseudo inverse estimator within this framework, revealing a bias-variance trade-off introduced by the surrogate network. We demonstrate a tighter asymptotic variance bound compared to previous studies and propose an enhanced variance estimator outperforming the original estimator. (3) We apply the pseudo inverse estimator to a real experiment involving over 50 million users, demonstrating its effectiveness in detecting network interference when combined with the difference-in-means estimator. Our research aims to bridge the gap between theoretical literature and practical implementation, providing a solution for estimating TTE in the presence of network interference and unknown interference structures.
\end{abstract}


\section{Introduction}

A/B testing, or randomized experiments, are essential tools for evaluating the impact of new product features in online platforms \citep{saveski2017detecting,saint2019-ego-cluster,chen2024-optimized-covariance,deng2024-pair-data-network}. The primary objective of A/B testing is to estimate the total treatment effect (TTE), which quantifies the difference between a scenario where all experimental units receive the current treatment and a counterfactual scenario where they all receive a new treatment. Classical A/B testing relies on the stable unit treatment value assumption (SUTVA) \citep{rubin1990-SUTVA}, which assumes that the treatment assigned to one unit does not affect any other units. However, this assumption may not hold in many situations, particularly when network interference is present \citep{hudgens2008-partial-interference,aronow2017-exposure-mapping}. For instance, when a new feature is tested on a subset of users in WeChat, the largest social platform in China, its effects can potentially spread to other users through information and content sharing. Ignoring network interference can lead to misleading experimental results and undermine data-driven decision-making.

Numerous methods have been proposed to improve TTE estimation in the presence of network interference. For example, partitioning the network into clusters and randomizing treatment at the cluster level has been shown to reduce bias \citep{eckles2017-cluster-reduce-bias,holtz2024reducing-bias-cluster}. In the post-experiment phase, estimators such as regression-adjustment \citep{chin2019regression-adjustment,han2023regression-adjustment}, Horvitz-Thompson \citep{aronow2017-exposure-mapping}, and pseudo inverse estimators \citep{cortez2023-low-order-interaction,eichhorn2024-pseudo-inverse} have been developed to adjust for network interference. However, most of these methods assume that the network structure is known \textit{a priori} and limit interference to the 1-hop neighborhood. Additionally, assumptions made about potential outcome functions, such as linearity, low-order polynomial, or exposure mapping, are often not realistic in industrial applications. For example, in WeChat, experimenters may not know which units interfere with a specific unit due to evolving social relationships and interactions through common friends. Moreover, verifying these assumptions in the pre-experiment phase is challenging, increasing the risk of unreliable results. Therefore, it is crucial to bridge the gap between theoretical literature and practical implementation.

In this work, we focus on the pseudo inverse estimator, a method that has not been widely adopted in industry but exhibits promising theoretical properties. This estimator is applicable to both cluster-based and Bernoulli randomization designs and has been shown to have lower variance compared to the Horvitz-Thompson estimator \citep{eichhorn2024-pseudo-inverse}. We aim to investigate its performance under a broader and more practical-oriented setting. Our contributions are threefold: (1) We introduce the surrogate network framework, which models the practical scenario where experimenters construct an approximation of the true interference network using observable data. (2) We analyze the performance of the pseudo inverse estimator within this framework, demonstrating a tighter asymptotic variance bound compared to previous work, and propose an improved variance estimator that outperforms the original one. (3) We apply the pseudo inverse estimator to a real experiment with over 50 million users, showing that combining it with the difference-in-means estimator can effectively detect network interference.

The paper is structured as follows: In Section 2, we review related work. Section 3 presents our theoretical framework. In Section 4, we analyze the bias and variance of the estimator used. Section 5 discusses variance estimation and statistical inference results. We verify our theoretical results through a comprehensive simulation study in Section 6 and present an empirical study in a real experiment in WeChat in Section 7. Finally, we conclude in Section 8.

\section{Related works}

There are various types of interference effects that violate SUTVA, including carryover \citep{bojinov2023switchback}, spatial \citep{leung2022-spatial-cluster}, and network effects \citep{ugander2013-GCR}, among others. For a comprehensive review of interference, we refer readers to \cite{halloran2016interference-book}. While our work focuses on interference under a general network, there are also studies on bipartite networks \citep{brennan2022bipartite-linear,harshaw2023bipartite-linear} and random networks \citep{li2022random-network}, among others. Unlike most literature that assumes the interference network is known \textit{a priori}, we study the case when experimenters can only observe a surrogate network, which approximates the true network. A similar setting was studied by \cite{li2021network-noise}, who used method-of-moments estimators under the assumption that the observed network is generated from the true network through a random process. Another work on causal inference under network uncertainty is \cite{bhattacharya2020structure-learning}, which applied a structure learning approach. We also mention the analysis of misspecified exposure mapping \citep{savje2024misspecified-exposure}, which can be extended to the analysis of Horvitz-Thompson estimator under our setting.

In the pre-experiment phase, several experiment design approaches have been proposed to mitigate network interference, such as cluster-based randomization \citep{ugander2013-GCR}. Empirical evidence shows that cluster-based design can reduce bias when interference exists \citep{holtz2024reducing-bias-cluster}. It has been shown that there is a bias-variance trade-off in the design of clusters \citep{viviano2023causal-clustering}. Larger clusters usually mean smaller bias and larger variance, motivating the design of clustering algorithms for causal inference \citep{ugander2013-GCR,ugander2023-RGCR,viviano2023causal-clustering}. In addition to cluster-based design, combining cluster-based and Bernoulli randomization can also be used to tackle interference \citep{jiang2023mixed-design}. When a series of experiments is possible, staggered roll-out design is another option under network interference \citep{cortez2022staggered-low-order}.

Different estimators have been shown to have varying performance under different assumptions. \cite{chin2019regression-adjustment} demonstrates that the OLS estimator is consistent for TTE estimation given a homogeneous linear data generation process. In a network with $n$ nodes and maximum degree $d$, \cite{jiang2023mixed-design} proposed an estimator under heterogeneous linear potential outcome functions with an MSE of $O(d^3/(np))$, where $p\le 0.5$ is the marginal treatment probability of units. \cite{cortez2023-low-order-interaction} showed that the MSE of the pseudo inverse estimator is $O(d^{\beta+2}/(np^{\beta}))$, given polynomial potential outcome functions with maximum degree $\beta$. \cite{ugander2013-GCR} presented a $O(d^4/(np^d))$ bound on the MSE of the Horvitz-Thompson estimator under cluster-based design, which was later improved to $O(d^6/(np^d))$ in \cite{ugander2023-RGCR}. Our result can be used to show a $O(d^2/(np))$ bound under linear potential outcome functions, which, to the best of our knowledge, is the tightest bound under this setting.

Beyond estimating TTE, other research goals have attracted attention in the literature on network interference, such as estimating average direct effect \citep{savje2021average-direct-effect}, minimizing the worst-case variance of cluster-based design \citep{candogan2024correlated-cluster}, and testing for the existence of network interference \citep{saveski2017detecting,athey2018exact-p-value,han2023detecting-interference}. We have also proposed an approach for testing SUTVA without requiring specific experimental design or Monte-Carlo simulation.

\section{Setup}
The population consists of $n$ units. We denote the treatment assignment vector as $\vec z \in \{0,1\}^n$, where $z_i = 1$ indicates unit $i$ is assigned to the treatment group, and $z_i = 0$ if $i$ is assigned to the control group. Let $Y_i(\vec z)$ represent the potential outcome of unit $i$ under treatment assignment $\vec z$. A key estimand of interest is the Total Treatment Effect (TTE), defined as the difference in average outcomes when all or no units receive treatment:
\begin{align}
    \text{TTE} = \frac{1}{n} \sum_{i=1}^n \left[ Y_i(\vec 1) - Y_i(\vec 0) \right] \label{eq:TTE definition}
\end{align}

Identifying TTE is not feasible without restrictions on how $Y_i(\vec z)$ can vary with $\vec z$. The prevailing approach in the literature assumes that interference is represented by a dependency network $\mathcal{A}$. We consider $\mathcal{A}$ to be undirected and represent it as a binary symmetric matrix, where $A_{ij} = 1$ indicates that $Y_i$ is affected by the treatment assignment of unit $j$. By convention, $A_{ii} = 1$ for all $i = 1, 2, ..., n$. Let $\mathcal{N}_i = \{j: A_{ij} = 1\}$ denote the set of neighbors of unit $i$. $Y_i(\vec z)$ is solely a function of treatments within $\mathcal{N}_i$. It is important to note that we do not impose restrictions on the size of $|\mathcal{N}_i|$. We allow the neighborhood size $|\mathcal{N}_i|$ to be arbitrarily large, making this assumption versatile in practical applications. Furthermore, we maintain the standard assumption that potential outcomes are uniformly bounded:
\begin{assumption}[Bounded outcomes]
$|Y_i(\vec z)| \leq \bar{Y} < \infty, \quad \forall i = 1, 2, ..., n, \quad \vec{z} \in \{0,1\}^n.$\label{ass: bounded outcomes}
\end{assumption}

\subsection{Surrogate Network.}
Numerous studies adopting neighborhood interference assumptions implicitly rely on the assumption that the network $\mathcal{A}$ is known \textit{a priori}. However, this assumption is quite strong and often not feasible in many practical scenarios. In contrast, we argue that experimenters can only access a surrogate network $\mathcal{G}$, which may differ from the actual network $\mathcal{A}$. Similar to $\mathcal{A}$, $\mathcal{G}$ is represented by a binary symmetric matrix, and $G_{ij}$ is interpreted in the same manner. Let $\mathcal{M}_i = \{j: G_{ij} = 1\}$ denote the surrogate neighbor set of unit $i$.

Taking WeChat as an example, which runs hundreds of experiments involving network interference daily, the original social network comprises over one billion users and hundreds of billions of connections, leading to a substantial computational load. To reduce time and expenses, experimenters usually retain only edges that represent specific social interactions within the past 28 days, resulting in a sparser network. Such a sparse network, constructed from the underlying social relationships, can be considered a surrogate network according to our framework.

\subsection{Potential Outcomes}

To facilitate tractable inference of causal estimands, we consider the following type of potential outcome functions:

\begin{assumption}[Potential Outcome Function]
Let the potential outcome functions be denoted as $Y_i(\vec{z}) = f_i(\vec{z})$. $C_0$ is a universal constant that is independent of $\mathcal{A}$ and $\mathcal{G}$. Furthermore, let $\boldsymbol{W}=\{w_{ij}\}^{n\times n}$ be an unknown non-negative matrix such that $w_{ij}=0$ if $j\notin \mathcal{N}_i\;\forall i$, and both $||\boldsymbol{W}||_1\le C_0$ and $||\boldsymbol{W}||_{\infty}\le C_0$ hold. For all $i=1,2,...,n$, $f_i$ is unknown. Define $\psi_i^k(\vec z_{-\{k\}})=f_i(\vec z_{-\{k\}},z_k=1)-f_i(\vec z_{-\{k\}},z_k=0)$; then, we have $|\psi_i^k(\vec z_{-\{k\}})|\le C_0 w_{ik}$ and $|\psi_i^k(\vec z_{-\{k,l\}}, z_l=0)-\psi_i^k(\vec z_{-\{k,l\}}, z_l=1)|\le C_0 w_{ik}w_{il}$ for all $k\ne l$.
\label{ass: potential outcomes}
\end{assumption}

Here, $\psi_i^k(\vec z_{-\{k\}})$ represents the change in the potential outcome of unit $i$ when $z_k$ is switched from $0$ to $1$, given the treatment assignments $\vec z_{-k}$ to the remaining units. The matrix $\boldsymbol{W}$ has finite $1$ and $\infty$ norms. The constraint $||\boldsymbol{W}||_1\le C_0$ bounds the total variation of potential outcomes, ensuring that $TV(f_i)\le \sum_{k\in\mathcal{N}_i}\max_{\vec z_{-k}}|\psi_i^k(\vec z_{-\{k\}})|\le C_0\sum_{k\in\mathcal{N}_i} w_{ik}=O(1)$ for all $i$. Similarly, $||\boldsymbol{W}||_{\infty}\le C_0$ prevents any unit from exerting excessive "influence", ensuring that changes in $z_i$ result in only a finite total change in potential outcomes. We also impose an upper bound on the total variation of $\psi_i^k$ by $O(w_{ik})$. We aim to ensure that the potential outcomes remain sufficiently "smooth" regardless of changes in the dependency network $\mathcal{A}$.

To gain insight into the intuition behind this assumption, we examine the relationship between altering the recommendation algorithm and the daily time spent watching videos on the Wechat Channel. The modification of the recommendation algorithm directly impacts user behavior, while interference also arises from video sharing among users through social networks. The time a user spends watching Wechat Channel videos is related to their exposure, which is the frequency of encountering videos from either system recommendations or friends' shares. In this context, let $\vec{e}$ represent the exposure vector for all units, $\vec{z}$ indicate whether to change the recommendation for each user, and $\boldsymbol{\Delta}$ be a diagonal matrix where the $i$'th diagonal entry denotes the direct impact of the treatment. Under a well-established social interaction model:
\begin{align}
    \vec{e} = \boldsymbol{\Delta}\vec{z} + \boldsymbol{P}\vec{e}+\vec{\alpha}, \label{eq: linear exposure}
\end{align}
where $\boldsymbol{P}$ is the sharing probabilities matrix, an $n \times n$ stochastic matrix with diagonal entries set to $0$, and $\vec\alpha$ is the status quo. Under certain mild conditions, such as $||\boldsymbol{P}||_1<0.9$ and $||\boldsymbol{P}||_{\infty}<0.9$, $\vec e$ is linear in $\vec z$, that is, $\vec e=(\boldsymbol{I}-\boldsymbol{P})^{-1}(\boldsymbol{\Delta}\vec{z}+\vec\alpha)$, which can also be expressed as
\begin{align*}
    e_i=\alpha_i'+\sum_{j=1}^n w_{ij}z_j=\alpha_i'+\sum_{j\in \mathcal{N}_i}w_{ij}z_j,\;\forall i=1,2,...,n
\end{align*}
for a weight matrix $\boldsymbol{W}=(\boldsymbol{I}-\boldsymbol{P})^{-1}\boldsymbol{\Delta}=\{w_{ij}\}_{n\times n}$. It is easy to verify that $\boldsymbol{W}$ has bounded $1$ and $\infty$ norms. $\mathcal{N}_i$ is the set of $j$ for which $w_{ij}$ is nonzero. We assume that the time spent watching Wechat Channel videos is a function of exposure, therefore
\begin{equation}
    Y_i(\vec z)=Y_i(e_i)=f_i\left(\sum_{j\in \mathcal{N}_i}w_{ij}z_j\right) \label{eq: example model}
\end{equation}
for an unknown function $f_i:\mathbbm{R}\rightarrow\mathbbm{R}$. To relate this example to Assumption \ref{ass: potential outcomes}, we let $f_i$ be a differentiable function with an $L$-Lipschitz continuous and bounded first-order derivative. Consequently, we have
\[|\psi_i^k(\vec z_{-\{k\}})|= \left|\int_{0}^{w_{ik}}f'\left(\sum_{j\in \mathcal{N}_i\backslash\{k\}}w_{ij}z_j+y\right)dy \right| =O(w_{ik})\]
And
\begin{align*}
    &|\psi_i^k(\vec z_{-\{k,l\}}, z_l=1)-\psi_i^k(\vec z_{-\{k,l\}}, z_l=0)|\\
    \le &\int_{0}^{w_{ik}}\left|f'\left(\sum_{j\in \mathcal{N}_i\backslash\{k,l\}}w_{ij}z_j+w_{il}+y\right)dy -f'\left(\sum_{j\in \mathcal{N}_i\backslash\{k,l\}}w_{ij}z_j+y\right)\right|dy\\
    =&O(w_{ik}w_{il}),
\end{align*}
which satisfies Assumption \ref{ass: potential outcomes}. Although this example oversimplifies the real-world data generation process, it demonstrates that our assumptions can accommodate complex interference patterns.

\subsection{Design of Experiment.} Throughout this paper, we concentrate on the \textit{Uniform Bernoulli Design}, wherein each unit is independently assigned to treatment with a uniform probability $p \in (0,1)$. This experimental design is prevalent in standard A/B testing, known for its simplicity and implementation ease. It is extensively utilized in WeChat, with thousands of experiments conducted daily. This approach facilitates our re-analysis of existing experiment data, enabling us to adjust for interference without the need to initiate a new experiment, which could be time-consuming and costly.

\section{Estimator}\label{section: estimator}

In this section, we investigate the pseudo inverse estimator, proposed by \citep{eichhorn2024-pseudo-inverse}, and alternatively known as the SNIPE estimator \citep{cortez2023-low-order-interaction}, within the framework of a surrogate network. For practical applications, we assign a value of one to the low-order parameter (see Remark \ref{remark: low-order parameter}).

\begin{align}
    \hat{\tau}(\mathcal{G})=\frac{1}{n} \sum_{i=1}^n Y_i\sum_{j\in \mathcal{M}_i}\left(\frac{z_j}{p}-\frac{1-z_j}{1-p}\right)\label{eq: estimator}
\end{align}

\cite{cortez2023-low-order-interaction} demonstrated that, under the assumption of linear potential outcomes, $\hat{\tau}(\mathcal{A})$ is an unbiased estimator for the TTE, with a variance of Var$(\hat{\tau}(\mathcal{G}))=O\left(d_{\mathcal{A}}^3/ np(1-p)\right)$, where $d_{\mathcal{A}}$ represents the maximum degree of the underlying true network $\mathcal{A}$. In our context, due to the experimenter's inability to fully observe the true network $\mathcal{A}$, the original estimator is constrained to the surrogate network $\mathcal{G}$. Moreover, Assumption \ref{ass: potential outcomes} does not require the potential outcome to have a linear or polynomial form, which make theoretical reanalysis necessary. In this section, we will show new theoretical properties of the estimator under our refined assumptions, which offers new insights into industrial implementation. We first analyze the bias under the refined assumption, then we derive an asymptotic variance upper bound that relies on the maximum degree of $\mathcal{G}$, yielding a tighter bound than the one proposed in the original paper.

\begin{remark}[Low-order Parameter]
    \cite{cortez2023-low-order-interaction} established that, when the potential outcome function is of degree at most $\beta$, employing a SNIPE (pseudo inverse) estimator with a low-order parameter $\beta$ results in an MSE that can be upper-bounded by $O\left(\frac{d^{\beta+2}}{np^{\beta}(1-p)^{\beta}}\right)$, where $d$ is the maximum degree of the network. The reason for this article to focus on $\beta=1$ is twofold: (1) The worst-case variance may grow exponentially with $\beta$, leading to a loss of statistical power for the estimator. For instance, the variance when $\beta=2$ can be hundreds of times greater than when $\beta=1$. (2) The computational complexity associated with the SNIPE estimator is $O(n d^{\beta})$ for small $\beta$, which can render the estimation process time-consuming and even impractical in the context of large social networks.
    \label{remark: low-order parameter}
\end{remark}

Under our new assumptions about the surrogate network and potential outcomes, the pseudo inverse estimator does not necessarily provide an unbiased estimation. To see the reason behind, we first check the expected value:
\begin{lemma}\label{lemma: expectation of estimator}Under Assumption \ref{ass: potential outcomes}, the expected value of the proposed estimator is
\[E(\hat{\tau}(\mathcal{G}))=\frac{1}{n}\sum_{i=1}^n\sum_{k\in\mathcal{M}_i}E(\psi_i^k(\vec z_{-\{k\}}))\]
\end{lemma}
\begin{proof}
    See Appendix \ref{appendix: proof of expectation}.
\end{proof}
Here $E(\psi_i^k(\vec z_{-\{k\}}))$ is the expected marginal increment of $i$'s potential outcome given that $z_k$ is switched from $0$ to $1$, which can be viewed as a linear approximation to the interference from $k$ to $i$. Additionally, due to the missed edges in the $\mathcal{G}$, the interference outside the surrogate neighborhood $\mathcal{M}_i$ is ignored. This results in two types of bias, one from the linear approximation and the other from the surrogate network $\mathcal{G}$. We will discuss the two types of bias, called endogenous and exogenous bias, in the following subsections. The next corollary provides a sufficient condition under which the two types of bias does not exist, which is exact the same assumption as in \cite{cortez2023-low-order-interaction}.
\begin{corollary}\label{corollary: unbiased estimator}
$\hat{\tau}(\mathcal{G})$ is unbiased if $\mathcal{A}=\mathcal{G}$ and the potential outcomes are linear in $\vec z$.
\end{corollary}

\subsection{Bias} The exogenous bias is from the mismatch between the ground truth network $\mathcal{A}$ and the surrogate network $\mathcal{G}$. The missing edges in $\mathcal{G}$ can make the estimator underestimate the interference, thereby causing bias. To give a quantitative explanation, consider a popular linear model:
\begin{equation}
    Y_i(\vec z)= Y_i(\vec 0)+\sum_{j\in \mathcal{N}_i}w_{ij} z_j\;\forall i,\label{eq: linear potential outcome}
\end{equation}
where $\boldsymbol{W}$ is defined in Assumption \ref{ass: potential outcomes}. The following assumption is used to quantify the difference between $\mathcal{A}$ and $\mathcal{G}$.

\begin{assumption} [Gap to the ground truth] \label{ass: gap between A, G}
There exist $\delta\in [0,1]$ such that 
    \begin{align*}
        \sum_{j=1}^n w_{ij}A_{ij}(1-G_{ij})\le\delta\sum_{j=1}^n w_{ij},\;\forall i,
    \end{align*}
\end{assumption}
where $A_{ij}(1-G_{ij})$ means edge $(i,j)$ is in the true but not in the surrogate network. The relative bias in this scenario is related to $\delta$, which is explained as the relative weighted difference between $\mathcal{A}$ and $\mathcal{G}$.
\begin{lemma}\label{lemma: exogenous bias}
Under the potential outcomes (\ref{eq: linear potential outcome}) and Assumption \ref{ass: gap between A, G}, the relative bias of the estimator (\ref{eq: estimator}) is $O(\delta)$, i.e.
    \[\frac{|E(\hat{\tau}(\mathcal{G}))-\text{TTE}|}{|TTE|}= O(\delta)\]
\end{lemma}
\begin{proof}
    See Appendix \ref{appendix: proof of lemma exogenous bias}.
\end{proof}
For the case of nonlinear potential outcome, we can show that the absolute value of the exogenous bias is bounded by $O(\delta)$ under Assumption \ref{ass: potential outcomes}. The proof is trivial and omitted here. 

The endogenous bias of $\hat\tau(\mathcal{G})$ is from the non-linearity of potential outcomes. Without more information about $f_i$ except for Assumption \ref{ass: potential outcomes}, we are not able to give a quantitative explanation in terms of $\boldsymbol{W}$. But follows from \cite{cortez2023-low-order-interaction}, we can give a qualitative explanation. With some abuse of notation, we equivalently present $f_i$ as $f_i(S)\equiv f_i(\vec z=\{\mathbbm{1}\{i\in S\}\}_{i=1}^n)$, in which the input is changed from a vector to a subset $S$ of $\{1,...,n\}$. Then we can rewrite $f_i$ as a polynomial function:
\begin{align*}
    f_i(\vec z)=\sum_{S\subseteq \mathcal{N}_i}f_i(S)\prod_{i\in S}z_i\prod_{j\in \mathcal{N}_i\backslash S}(1-z_j)=\sum_{S\subseteq \mathcal{N}_i}a_{i,S}\prod_{k\in S}z_k
\end{align*}
where $a_{i,S}=\sum_{S'\subseteq S}f_i(S')(-1)^{|S\backslash S'|}$. We call $a_{i,S}$ the joint treatment effect of $S$ for unit $i$. Define the $\beta$'th-order joint treatment effect as $\bar a_{\beta}=\frac{1}{n}\sum_{i=1}^n\sum_{S\subseteq\mathcal{N}_i:|S|=\beta} a_{i,S}$, 
then the TTE can be alternatively presented as $\text{TTE}=\sum_{\beta=0}^{d_{\mathcal{A}}} \bar a_{\beta}$.
The following Lemma shows that the endogenous bias of estimator $\hat\tau(\mathcal{G})$ is from the underestimate of high-order joint treatment effect:
\begin{lemma}\label{lemma: endogenous bias}
    When $\mathcal{A}=\mathcal{G}$, $\text{TTE}-E(\hat{\tau}(\mathcal{G}))= \sum_{\beta=0}^{d_{\mathcal{A}}}(1-\beta p^{\beta-1})\bar a_{\beta}$
\end{lemma}
\begin{proof}
    See Appendix \ref{appendix: proof of endogenous bias}.
\end{proof}
when $p=0.5$, $\hat\tau(\mathcal{G})$ can correctly estimate the first and second-order joint treatment effect, but underestimate the third-order one by $25\%$ and the forth-order one by $50\%$, et al. Smaller $p$ will usually result in higher bias, thus we recommend to use $p=0.5$ in implementation.

We believe that the above analysis towards bias provides useful insights into practice, since in the most cases, $\mathcal{A}$ does not equal to $\mathcal{G}$, and the potential outcomes might deviate from linear. As a practical recommendation, we suggest experimenters to using historical data to verify the constructed surrogate network before the experiment, and avoid the scenario under which high-order effect may be significant.  

\subsection{Variance}
In this section, we investigate the asymptotic behavior of $\Var( \hat{\tau}(\mathcal{G}))$. We first derive the asymptotic upper bound as a function of $d_{\mathcal{G}}$, $n$ and $p$, where $d_{\mathcal{G}}$ denote the largest degree of network $\mathcal{G}$.
The following theorem summarizes the key theoretical insights of this article, which can be used to guide the choice of sparsity when we design the surrogate network.

\begin{theorem}[Variance Upper Bound]\label{theorem: variance upper bound}
Under Assumption \ref{ass: bounded outcomes}$\sim$ \ref{ass: potential outcomes}, the estimator defined in (\ref{eq: estimator}) has the following asymptotic variance upper bound:
\[\Var( \hat{\tau}(\mathcal{G}))=O\left(\frac{d_{\mathcal{G}}^2}{np(1-p)}\right)\]
\end{theorem}
\begin{proof}
    See Appendix 
\end{proof}
The proof idea is to first rewrite the variance as   
\[\frac{1}{n^2}\sum_{i=1}^n\sum_{j=1}^n \sum_{k\in \mathcal{M}_i}\sum_{l\in \mathcal{M}_j}\Cov(Y_iD_k,Y_jD_l),\]
where $D_i=\left(\frac{z_i}{p}-\frac{1-z_i}{1-p}\right)$, and then derive a bound as a function of $\boldsymbol{W}$ for $|\Cov(Y_iD_k,Y_jD_l)|$ under two different cases of $k=l$ and $k\ne l$. Then we use the assumption on the $1$ and $\infty$ norms of $\boldsymbol{W}$ to bound the summation of $|\Cov(Y_iD_k,Y_jD_l)|$. Our second result provides asymptotic lower bounds.
\begin{theorem}[Variance Lower Bound]\label{theorem: variance lower bound}
Let $\mathcal{G}$ be a $d_{\mathcal{G}}$-regular network. Suppose all potential outcomes are a constant, which complies with Assumptions \ref{ass: bounded outcomes}$\sim$ \ref{ass: potential outcomes}. The variance of the estimator defined in (\ref{eq: estimator}) exhibits the following lower bound:
\[\Var( \hat{\tau}(\mathcal{G}))=\Omega\left(\frac{d_{\mathcal{G}}^2}{np(1-p)}\right)\]
\end{theorem}
\begin{proof}
    See Appendix \ref{appendix: proof of var lower bound}.
\end{proof}
The lower bound shows that we can construct potential outcomes and surrogate networks satisfying the assumption of Theorem \ref{theorem: variance upper bound} such that the variance is at least order $\Omega\left(\frac{d_{\mathcal{G}}^2}{np(1-p)}\right)$. Therefore, the variance upper bound in Theorem 1 is tight. The value of our theoretical result on the variance is twofold:
\begin{enumerate}
    \item The result implies that the variance primarily depend on the degree of $\mathcal{G}$, while the structure of $\mathcal{A}$ contributes at most a constant factor. This enables bias-variance trade-off in practice: incorporating more edges in $\mathcal{G}$ can reduce the exogenous bias at the cost of a higher variance, and vice versa.
    \item We obtain a stronger theoretical guarantee compared with \cite{cortez2023-low-order-interaction} and \cite{eichhorn2024-pseudo-inverse}. They obtain a $O\left(\frac{d_{\mathcal{G}}^3}{np(1-p)}\right)$ bound under the linear potential outcome and require $\mathcal{A}=\mathcal{G}$. Our result improve the numerator in the upper bound from $d_{\mathcal{G}}^3$ to $d_{\mathcal{G}}^2$ under weaker assumptions. And we show that our bound is tight and can not be further improved.
\end{enumerate}

We provide simulation result to verify our theoretical findings on the variance bound. Furthermore, we numerically compare the empirical bias and variance against the cluster based design, under synthetic networks in Section \ref{section: experiment}. 

\section{Inference}\label{section: inference}
In this section, we improve the variance estimation in \cite{cortez2023-low-order-interaction} by a much more efficient variance estimator. We first state assumption for asymptotic inference. Define $\sigma_{\mathcal{G}}^2=\Var(\hat\tau(\mathcal{G}))$.
\begin{assumption}[Non-degeneracy]\label{ass: var non degenerate}
    \[ \lim\inf_{n\rightarrow\infty} \ n\sigma_{\mathcal{G}}^2\backslash d_{\mathcal{G}}^2>0\]
\end{assumption}
This is a standard condition and reasonable to impose in light of the bounds on the variance derived in Theorem \ref{theorem: variance upper bound} and \ref{theorem: variance lower bound}.

\subsection{Variance Estimator}
Our insights for the variance estimator are from Theorem 4 in \citet{leung2022-spatial-cluster}. Define $T_{ij}=Y_i\left(\frac{z_j}{p}-\frac{1-z_j}{1-p}\right)$, $T_{i}=\sum_{j\in \mathcal{M}_i}Y_{ij}$ and $I_{ij}=\mathbbm{1}\{\mathcal{M}_i\cap\mathcal{M}_j\ne\emptyset\}$. Our proposed variance estimator is
\begin{align}\label{eq: variance estimator}
    \hat\sigma_{\mathcal{G}}^2=\frac{1}{n^2}\sum_{i=1}^n\sum_{j=1}^n ( T_i-\hat\tau(\mathcal{G}))( T_j-\hat\tau(\mathcal{G}))I_{ij},
\end{align}

For the ease of theoretical analysis, we impose another assumption on the potential outcome functions.
\begin{assumption}\label{ass: Lipchitz second derivative Y_i}
    Define $\phi_i^{kl}(\vec z_{-\{k,l\}})=\psi_i^k(\vec z_{-\{k,l\}}, z_l=0)-\psi_i^k(\vec z_{-\{k,l\}}, z_l=1)$, then
    \begin{align*}
        |\phi_i^{kl}(\vec z_{-\{k,l,j\}},z_j=1)-\phi_i^{kl}(\vec z_{-\{k,l,j\}},z_j=0)|\le C_0 w_{ik}w_{il}w_{ij},\; \forall k\ne l, l\ne j,j\ne k
    \end{align*}
\end{assumption}
This is another "smoothness" assumption regarding the potential outcome functions. To build intuition, reconsider the example given after Assumption \ref{ass: potential outcomes}. We claim that if $f_i$ have bounded and $L$-Lipschitz continuous second-order derivative, then Assumption \ref{ass: Lipchitz second derivative Y_i} is satisfied. The proof for this claim is simple and thus omitted. We believe that such an assumption is not restrictive and will not undermine the effectiveness of the proposed variance estimator.

 The next theorem is used to show the asymptotic property of this variance estimator

\begin{theorem}\label{theorem: var estimator bias}
Under assumptions \ref{ass: bounded outcomes} to \ref{ass: Lipchitz second derivative Y_i}, and assuming the treatment probability $p$ is fixed, as well as $d_{\mathcal{G}}^6=o(n)$, we have
\[\lim_{n\rightarrow\infty}\frac{n}{d_{\mathcal{G}}^2}(\hat\sigma_{\mathcal{G}}^2-\sigma_{\mathcal{G}}^2)\rightarrow O(\delta)+\mathcal{R}_{\mathcal{G}},\]
where
\[\mathcal{R}_{\mathcal{G}}=\frac{1}{nd_{\mathcal{G}}^2}\sum_{i=1}^n\sum_{j=1}^n [E( T_i)-E(\hat\tau(\mathcal{G}))][E( T_j)-E(\hat\tau(\mathcal{G}))]I_{ij},\]
and
\[E(T_i)=\sum_{k\in \mathcal{M}_i}E(\psi_i^k(\vec z_{-\{k\}})),\;\forall i\in \{1,...,n\}.\]
\end{theorem}

\begin{proof}
See Appendix \ref{appendix: proof of var estimator endogenous bias}.
\end{proof}

The proof follows the general outline provided in Theorem 4 of \citet{leung2022-spatial-cluster}, with the primary difference being the method used to derive the bound under our specific context. The assumption $d_{\mathcal{G}}^6=o(n)$ ensures that the constructed surrogate network remains sufficiently sparse. The bias term $O(\delta)$ arises from the surrogate network, indicating that missing edges may increase the likelihood of inaccurate variance estimation. The term $\mathcal{R}_{\mathcal{G}}$ is $O(1)$ and typically non-zero due to unit-level heterogeneity. For instance, under the conditions of Corollary \ref{corollary: unbiased estimator}, we have
\[\mathcal{R}_{\mathcal{G}}=\frac{1}{nd_{\mathcal{G}}^2}\sum_{i=1}^n\sum_{j=1}^n [Y_i(\vec 1)-Y_i(\vec 0)-\text{TTE}][Y_j(\vec 1)-Y_j(\vec 0)-\text{TTE}]I_{ij},\]
which usually does not approach zero asymptotically, except in the special case where treatment effects are homogeneous across all units, i.e., $Y_i(\vec 1)-Y_i(\vec 0)=\text{TTE}$ for all $i\in \{1,...,n\}$. As noted in \cite{leung2022-spatial-cluster}, this bias term is analogous to the bias present in the Neyman conservative estimator for the variance of the difference-in-means estimator. It is well-known that achieving a consistent estimation of the variance is not feasible in this context.

It is important to mention that the original work by \citep{cortez2023-low-order-interaction} utilized a variance estimator from \cite{aronow2017-exposure-mapping}, which was shown to be hundreds of times greater than the empirical variance in numerical simulations. Such an estimator lacks sufficient statistical power for practical use. In Section \ref{section: experiment}, we will provide numerical evidence to support the effectiveness of our proposed estimator.

\subsection{Hypothesis Testing}\label{section: Hypothesis Testing}

In this section, we demonstrate how to use the pseudo inverse estimator for testing different hypotheses. Practitioners are often interested in knowing whether their treatment leads to a change in the TTE and whether network effects are present in their experiment.

\paragraph{Testing Total Treatment Effect.}
We first explore methods for rejecting the null hypothesis that the TTE is zero. A conservative approach is to use Chebyshev's inequality, which states
\begin{align*}
    P(|\hat\tau(\mathcal{G})-E(\hat\tau(\mathcal{G}))|>k\sigma_{\mathcal{G}})\le \frac{1}{k^2},
\end{align*}
for any real number $k>0$. By rejecting the null hypothesis when $|\hat\tau(\mathcal{G})|>\sigma_{\mathcal{G}}/\sqrt{\alpha}$, the type-I error of our test is guaranteed to be no greater than $\alpha$. A less conservative approach assumes that $(\hat\tau(\mathcal{G})-E(\hat\tau(\mathcal{G})))/\sigma_{\mathcal{G}}$ follows a standard normal distribution, allowing us to construct the $(1-\alpha)\times 100\%$ confidence interval
\[(\hat\tau(\mathcal{G})+ \sigma_{\mathcal{G}}z_{\alpha/2}, \hat\tau(\mathcal{G})+ \sigma_{\mathcal{G}}z_{1-\alpha/2}),\]
where $z_{\alpha/2}$ and $z_{1-\alpha/2}$ are the $\alpha/2$ and $1-\alpha/2$ quantiles of the standard normal distribution. The following lemma establishes the asymptotic normality when the dependency network $\mathcal{A}$ has a bounded degree:

\begin{lemma}[Asymptotic Normality]
Under assumptions \ref{ass: bounded outcomes} to \ref{ass: potential outcomes} and \ref{ass: var non degenerate}, assuming the degree of the dependency network $d_{\mathcal{A}}$ is $O(1)$ and the treatment probability $p$ is fixed, $(\hat\tau(\mathcal{G})-E(\hat\tau(\mathcal{G})))/\sigma_{\mathcal{G}}$ converges in probability to a standard normal random variable as $n\rightarrow \infty$.
\end{lemma}
\begin{proof}
See the proof for Theorem 3 in \cite{cortez2023-low-order-interaction}.
\end{proof}

The proof relies on a well-established central limit theorem based on Stein's method, which requires $d_{\mathcal{A}}$ to be $O(1)$. While there is no off-the-shelf central limit theorem that directly applies to the surrogate network setting, we find that the normal approximation performs well in practice, even when the assumptions are violated. We provide further evidence of this through simulation in Section \ref{section: experiment}.

\paragraph{Testing Network Interference.}
Testing for network interference is essential for social platforms, as it can lead to inaccurate results in traditional A/B testing. Therefore, a crucial task is to test the null hypothesis of SUTVA. Note that the difference-in-means estimator
\[\hat\tau_{DIM}=\frac{1}{n}\sum_{i=1}^nY_i\left(\frac{z_i}{p}-\frac{1-z_i}{1-p}\right)\]
is equivalent to our estimator when $\mathcal{M}_i=\{i\}$. Based on Lemma \ref{lemma: expectation of estimator}, the expectation of our estimator under SUTVA is the same as the expectation of the difference-in-means estimator. This inspires us to combine the two estimators to test the null hypothesis of SUTVA. Similarly to Lemma \ref{lemma: expectation of estimator}, we can show
\[E(\hat{\tau}(\mathcal{G})-\hat\tau_{DIM})=\frac{1}{n}\sum_{i=1}^n\sum_{k\in\mathcal{M}_i\backslash\{i\}}E(\psi_i^k(\vec z_{-\{k\}})),\]
which equals zero under the null hypothesis of SUTVA. To estimate the variance of this new estimator, we use a variance estimator analogous to (\ref{eq: variance estimator}), replacing $T_i$ with $T_i'=\sum_{j\in \mathcal{M}_i\backslash\{i\}}T_{ij}$ and $\hat{\tau}(\mathcal{G})$ with $\hat{\tau}(\mathcal{G})-\hat\tau_{DIM}$. All subsequent analysis for (\ref{eq: variance estimator}) applies here as well. In practice, we find this approach to be effective in testing for the existence of network interference. We present our empirical findings in the next section.

\section{Experiments}\label{section: experiment}
There are four goals for this section. First, to investigate the variance bound in Theorem \ref{theorem: variance upper bound}. Second, to validate the asymptotic normal distribution under the surrogate network setting. Third, to compare our approach with difference-in-means estimator under both cluster-based and Bernoulli randomization. It is important to note that the pseudo inverse estimator is guaranteed to exhibit lower variance compared to the Horvitz-Thompson estimator \citep{eichhorn2024-pseudo-inverse}, which is omitted from the simulations for this very reason. Lastly, this section seeks to explore the empirical performance of our approach with a real-world experiment.

\subsection{Verification of theoretical results}\label{section: verify results}
 \textbf{Test Instances: }
 
 We let surrogate network $\mathcal{G}$ be a Erdős–Rényi network, which was chosen uniformly from the collection of all graphs which have $n$ nodes and $n\bar d$ edges. We interpret $\bar d$ as the average degree of $\mathcal{G}$. We adhere to the model presented in the example following Assumption \ref{ass: potential outcomes}, where the potential outcomes are defined according to (\ref{eq: example model}). We generate $\vec\alpha$ from i.i.d U$(0,1)$ distribution, the diagonal matrix $\boldsymbol{\Delta}$ with each diagonal entry drawn from a mutually independent U$(0,\gamma_1)$ distribution, and the stochastic matrix $\boldsymbol{P}=\{\gamma_2G_{ij}/\sum_{k}G_{ik}\}_{n\times n}$. Herein $\gamma_1$ represents the maximum direct treatment effect, and $\gamma_2$ denotes the sharing probability. This model naturally creates a dependency network $\mathcal{A}$ that diverges from $\mathcal{G}$, serving as a tool to verify our theoretical findings.
 
\textbf{Verify Theorem \ref{theorem: variance upper bound}: }

We define the potential outcome function as $Y_i=f(e_i)$, where $\vec{e}$ is derived from (\ref{eq: linear exposure}). We simulate the empirical variance of our estimator through 1000 replications, varying the choice of $\bar{d}$. The test configurations are set at $n=10000$, $p=0.5$, $\gamma_1=1$, $\gamma_2=0.5$, and $\bar{d}\in\{10,20,30,40,50,60,70,80,90,100\}$. We examine two distinct outcome functions. The first is a continuous function $f(x)=\sqrt{x}$, yielding a TTE$\approx0.3$. The second is a binary function $f(x)=\mathbbm{1}\{x>1\}$ with a TTE$\approx0.5$. We plot the empirical variance against the square of the average degree, $\bar{d}^2$. The findings are illustrated in Figure \ref{fig: verify variance}.

\begin{figure}[htbp]
    \centering 
    \begin{subfigure}{0.45\textwidth} 
        \includegraphics[width=\textwidth]{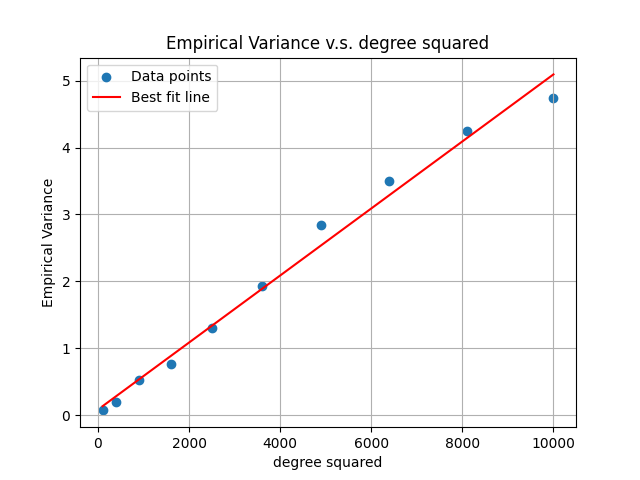}
        \caption{$f(x)=\sqrt{x}$} 
    \end{subfigure}
    \hspace{0.05\textwidth} 
    \begin{subfigure}{0.45\textwidth} 
        \includegraphics[width=\textwidth]{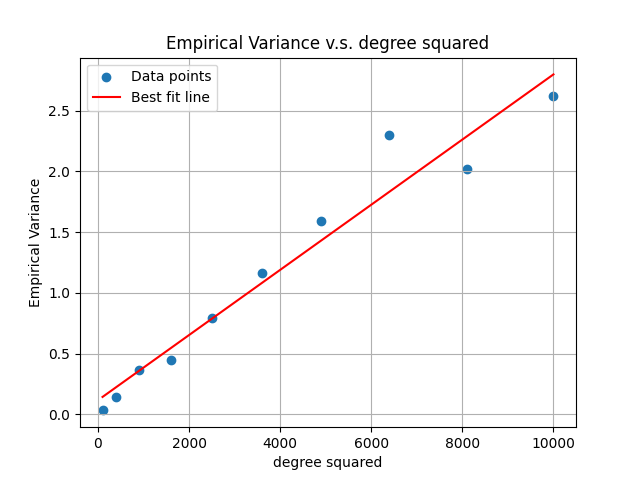}   
        \caption{$f(x)=\mathbbm{1}\{x>1\}$}
    \end{subfigure}
    \caption{Scatter plot of empirical variance v.s. $\bar d^2$} 
    \label{fig: verify variance} 
\end{figure}

In Figure \ref{fig: verify variance}, we add a best-fit line to ascertain whether a linear relationship exists between the empirical variance and the average degree $\bar{d}$. The results indicate that both scenarios exhibit a clear linear correlation, thereby confirming the accuracy of our variance bound.

\textbf{Verify approximate normality: }

We conduct simulations on the test instances with $f(x)=\mathbbm{1}\{x>1\}$, $n=10000$, $p=0.5$, $\gamma_1=1$, $\gamma_2=0.5$, and $\bar{d}\in\{10,20,30,40\}$ to assess the estimator's distribution for approximate normality. After obtaining 10,000 replications for each instance, we normalize the outcomes by their respective means and standard deviations. We plot the histogram of the normalized results against the density of a standard normal distribution in Figure \ref{fig: approximate normality}.

\begin{figure}[htbp]
    \centering 
    \begin{subfigure}{0.45\textwidth} 
        \includegraphics[width=\textwidth]{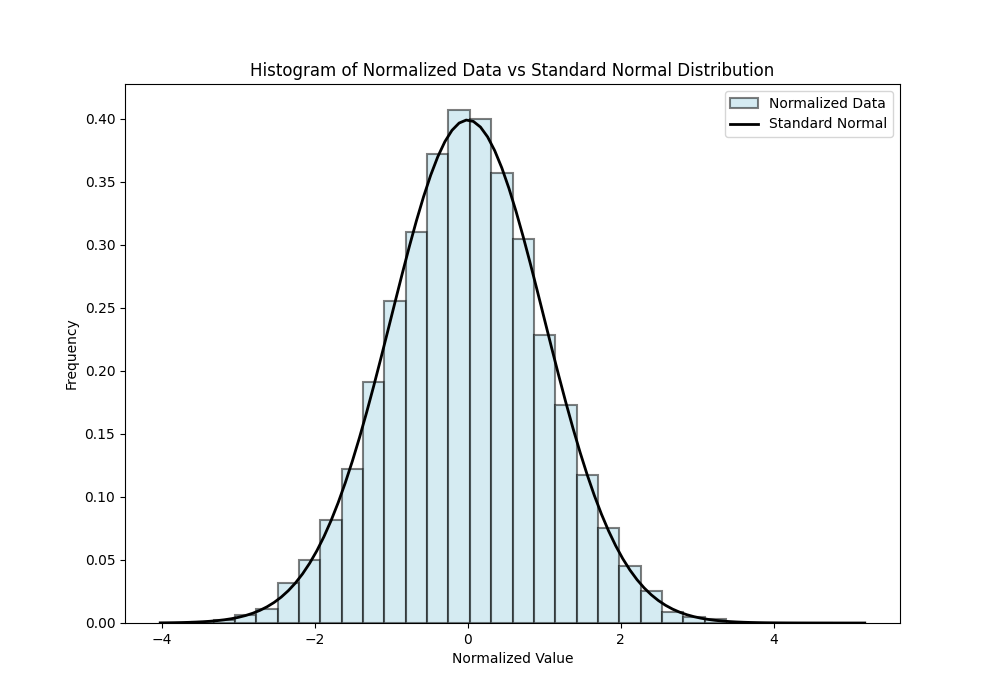} 
        \caption{$\bar d=10$} 
        \end{subfigure}
    \hspace{0.05\textwidth} 
    \begin{subfigure}{0.45\textwidth} 
        \includegraphics[width=\textwidth]{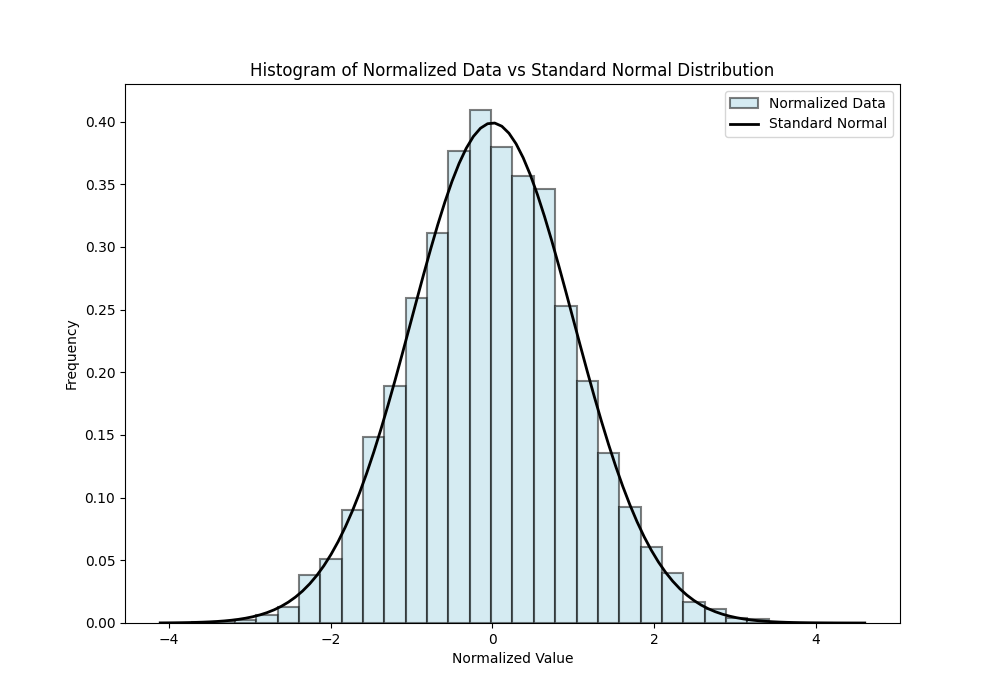} 
        \caption{$\bar d=20$} 
    \end{subfigure}
\\ 
    \begin{subfigure}{0.45\textwidth}
        \includegraphics[width=\textwidth]{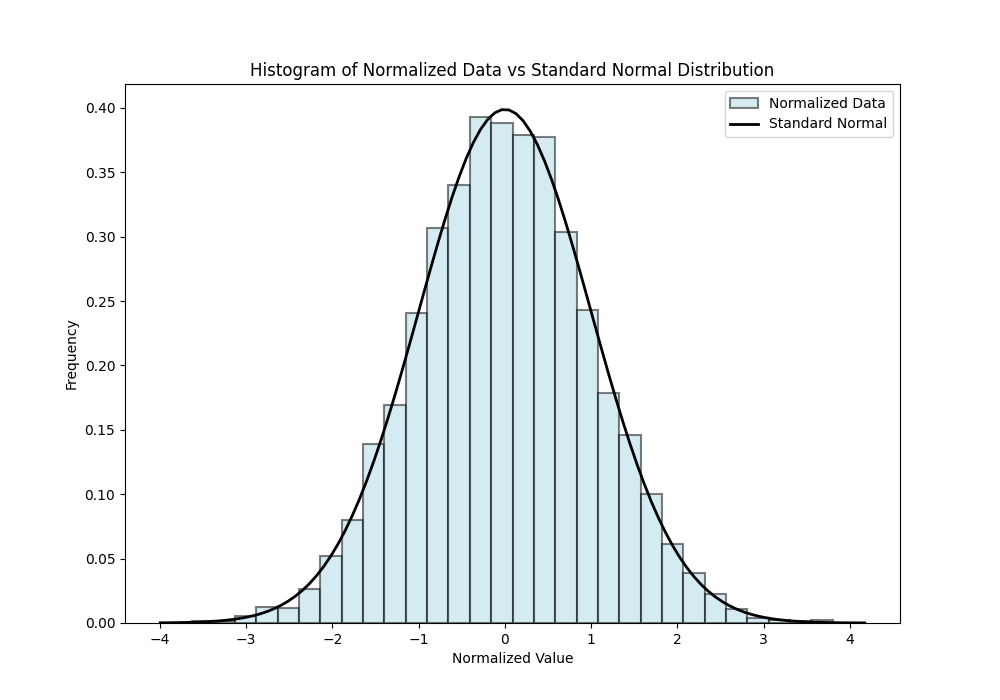} 
        \caption{$\bar d=30$}
    \end{subfigure}
    \hspace{0.05\textwidth} 
    \begin{subfigure}{0.45\textwidth} 
        \includegraphics[width=\textwidth]{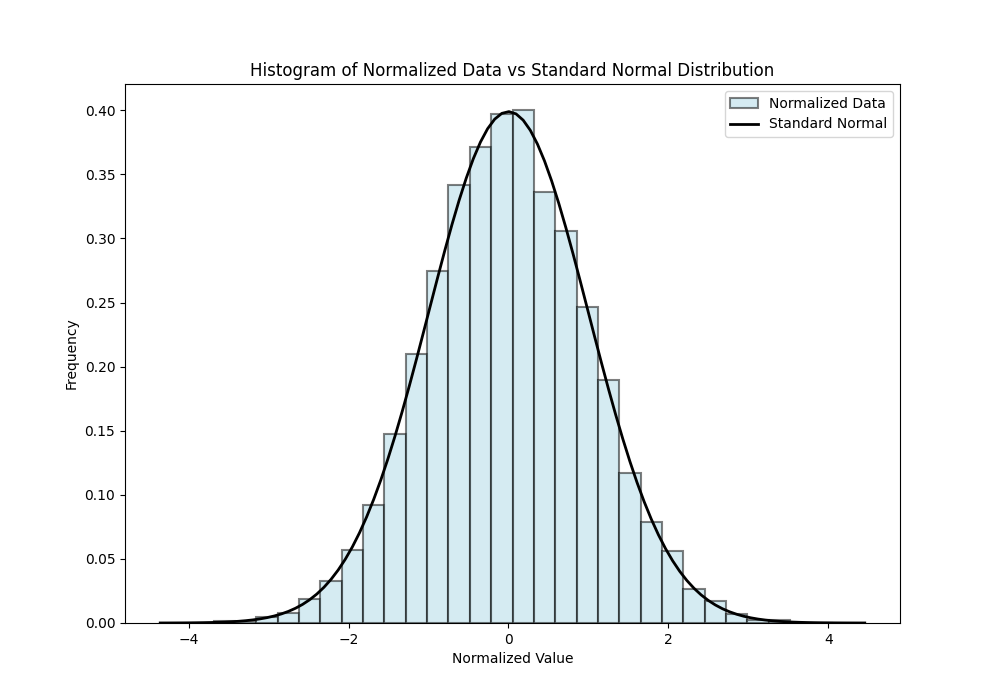}
        \caption{$\bar d=40$}
    \end{subfigure}
    \caption{Verification of approximate normality under different degrees}
    \label{fig: approximate normality}
\end{figure}

The simulations show that the estimator follows an approximately normal distribution under the surrogate network condition, provided that $\bar{d}$ is sufficiently small relative to $n$. Consequently, we can construct confidence intervals based on normal percentiles.

\textbf{Verify variance estimator: }

We compare the empirical variance with the estimated variance, computed using (\ref{eq: variance estimator}). We perform 1,000 simulations on test instances with identical parameters $n=10000$, $p=0.5$, $\gamma_1=1$, $\gamma_2=0.5$, but vary $f$ and $\bar{d}$. We calculate the mean and standard deviation of our variance estimator to identify the magnitude of potential bias. We denote the empirical variance, estimated variance, and the standard deviation of the estimated variance as $\sigma^2$,  $\hat\sigma^2$, and std$(\hat\sigma^2)$, respectively. The results are compiled in Table \ref{table: variance estimator}.
\renewcommand{\arraystretch}{1.5}
\begin{table}[htbp]
  \centering
  \caption{The performance of the proposed variance estimator}
    \begin{tabular}{c c c c c c c c c c}
      \toprule
      \multirow{2}{*}{$f(x)$} & \multicolumn{3}{c}{$\bar d=10$}& \multicolumn{3}{c}{$\bar d=20$}& \multicolumn{3}{c}{$\bar d=40$}\\
      \cline{2-10}
        & $\sigma^2$& $\hat\sigma^2$& std$(\hat\sigma^2)$& $\sigma^2$& $\hat\sigma^2$& std$(\hat\sigma^2)$& $\sigma^2$& $\hat\sigma^2$& std$(\hat\sigma^2)$\\
      \hline
      $\sqrt{x}$ & 0.07674  & 0.07606  & 0.00232  &0.24839   & 0.25755  & 0.00817  & 0.97396  & 0.84913  &  0.02718  \\
      $\mathbbm{1}\{x>1\}$ & 0.03837 & 0.03740 & 0.00128 &0.13519 & 0.13104&0.00484 &0.50640 &0.42184 & 0.01626 \\
      \bottomrule
    \end{tabular}\label{table: variance estimator}
\end{table}

Table \ref{table: variance estimator} reveals two insights. Firstly, the bias of our variance estimator escalates with the average degree $\bar{d}$. When $\bar{d}$ is considerably smaller than $n$, such as $\bar{d}=10$ and $\bar{d}=20$, the relative bias remains small, at approximately $2.5\%$ and $4\%$, respectively. However, as $\bar{d}$ increases to $40$, the bias rises to roughly $17\%$. This observation aligns with our theoretical finding in Theorem \ref{theorem: var estimator bias}, which suggests that the degree of the surrogate network should remain small. Practitioners can control the degree of the surrogate network to make the bias negligible. Secondly, the standard deviation of the estimated variance is relatively small, indicating that our variance estimator is reliable and stable in practical applications.

\subsection{Comparison between estimators}
We construct our test network from a subset of users residing in a specific region who have engaged in sharing behavior within the past 28 days. This network includes 100,015 nodes and 2,240,266 edges, where each node represents an individual and each edge signifies the presence of information sharing. The potential outcome model used here aligns with the one in Section \ref{section: verify results}, with parameters set to $n=10000$, $p=0.5$, $\gamma_1=1$, and $\gamma_2=0.5$.

We compare the pseudo inverse estimator with the difference-in-means estimator under Bernoulli and cluster-based randomization. The difference-in-means estimator, applied under Bernoulli randomization, serves as a benchmark in our simulation study, as it is commonly used to estimate the direct treatment effect. For cluster-based randomization, we employ a community detection algorithm known as Leiden \citep{leiden}, which generates 828 clusters. Our interest lies in determining which approach has better performance. We compare the bias, empirical variance ($\sigma^2$) and the mean square error (MSE) of three approaches under binary and continuous potential outcomes utilizing 1000 replication. The result is summarized in Table \ref{table: compare estimators}.
\renewcommand{\arraystretch}{1.5}
\begin{table}[htbp]
  \centering
  \caption{The performance of three different estimators}
    \begin{tabular}{c c c c c c c c c c}
      \toprule
      \multirow{2}{*}{$f(x)$} & \multicolumn{3}{c}{Difference-in-means Estimator}& \multicolumn{3}{c}{Cluster-based randomization}& \multicolumn{3}{c}{Pseudo Inverse Estimator}\\
      \cline{2-10}
        & Bias & $\sigma^2$& MSE & Bias & $\sigma^2$& MSE& Bias & $\sigma^2$& MSE\\
      \hline
      $\sqrt{x}$ & 0.28188& $7.0399e^{-7}$  & 0.07945  &0.22878   & $3.7554e^{-6}$ &  \textbf{0.05234}  & \textbf{0.20256}  & 0.03638  & 0.07741    \\
      $\mathbbm{1}\{x>1\}$ & 0.21351 & $4.8521e^{-7}$ & 0.04559 & 0.17888& $7.8015e^{-6}$&\textbf{0.03199} &\textbf{0.09589} &0.03974 & 0.04893 \\
      \bottomrule
    \end{tabular}\label{table: compare estimators}
\end{table}

Table \ref{table: compare estimators} yields two observations. Firstly, the pseudo inverse estimator demonstrates the smallest bias for both continuous and binary potential outcomes. The reason behind the bias of the cluster-based randomization is the inability to perfectly partition the test network. Only 28.7\% of edges connect endpoints within the same cluster, leading to an underestimation of interference effects. Secondly, while the pseudo inverse estimator exhibits the highest variance—a consequence of its variance scaling with the squared degree of the network—it becomes more advantageous as the number of nodes increases under a constant average degree, given that the MSE becomes dominated by bias.

\subsection{Application }
We apply our approach to a comprehensive real-world experiment conducted within WeChat, involving 53,603,004 nodes and 1,066,143,998 edges. The experimental design employs uniform Bernoulli randomization with a probability of $p=0.5$. We calculate the difference-in-means estimator $\hat{\tau}_{dim}$, the pseudo inverse estimator $\hat{\tau}_{pi}$, and the difference $\hat{\tau}_{pi}-\hat{\tau}_{dim}$ across 11 metrics that could potentially be affected by network interference. To estimate the variance of $\hat{\tau}_{dim}$, we use the Neyman estimator, while the approach outlined in Section \ref{section: inference} is utilized to estimate the variance of $\hat{\tau}_{pi}$ and $\hat{\tau}_{pi}-\hat{\tau}_{dim}$. The results are presented in Table \ref{table: real experiment}, with each row corresponding to a specific metric.
\renewcommand{\arraystretch}{1.5}
\begin{table}[htbp]
  \centering
  \caption{Results from a real experiment}
    \begin{tabular}{c c c c c c c c c c}
      \toprule
      \multirow{2}{*}{} & \multicolumn{3}{c}{$\hat{\tau}_{dim}$}& \multicolumn{3}{c}{$\hat{\tau}_{pi}$}& \multicolumn{3}{c}{$\hat{\tau}_{pi}-\hat{\tau}_{dim}$}\\
      \cline{2-10}
        &Value & Est. Var.& $p$-value & Value & Est. Var.& $p$-value &Value & Est. Var.& $p$-value \\
      \hline
      1&$1.001e^{-3}$&$1.814e^{-6}$&0.4571&0.1353&$4.708e^{-3}$&0.0485&0.1343&$4.467e^{-3}$&0.0444\\
      2&0.4988&$4.582e^{-3}$&$1.712e^{-13}$&1.6131&0.6177&0.0401&1.1142&0.5871&0.1459\\
      3&21.237&1.8816&$0.0000$&47.090&659.35&0.0667&25.852&624.12&0.3007\\
4&0.0116&$4.993e^{-6}$&$1.994e^{-7}$&0.0463&$4.140e^{-4}$&0.0228&0.0346&$3.955e^{-4}$&0.0811\\
5&$-1.215e^{-2}$&$8.148e^{-6}$&$2.078e^{-5}$&$-4.600e^{-4}$&$9.845e^{-4}$&0.9883&$-1.976e^{-3}$&$9.485e^{-4}$&0.9488\\
6&$1.500e^{-3}$&$3.624e^{-7}$&0.0127&$-4.600e^{-4}$&$9.180e^{-5}$&0.9617&$-1.960e^{-3}$&$8.909e^{-5}$&0.8355\\
7&$4.940e^{-3}$&$9.054e^{-7}$&$2.084e^{-7}$&0.0180&$1.096e^{-4}$&0.0859&0.0130&$1.062e^{-4}$&0.2062\\
8&$-1.808e^{-2}$&$5.544e^{-6}$&$1.598e^{-14}$&$-2.704e^{-2}$&$4.091e^{-4}$&0.1811&$-8.960e^{-3}$&$3.952e^{-4}$&0.6522\\
9&$3.781e^{-3}$&$3.600e^{-5}$&0.5285&0.2882&0.0277&0.0834&0.2844&0.0263&0.0792\\
10&0.0445&$6.751e^{-5}$&$6.063e^{-8}$&0.2199&0.0239&0.1550&0.1754&0.0233&0.2512\\
11&0.0637&$1.493e^{-4}$&$1.890e^{-7}$&0.3536&0.0475&0.1049&0.2899&0.0465&0.1787\\
      \bottomrule
    \end{tabular}\label{table: real experiment}
\end{table}

Among the 11 metrics, we identify network interference in 1 metric at a $95\%$ confidence level and in 2 metrics at a $90\%$ confidence level, based on the $p$-value of $\hat{\tau}_{pi}-\hat{\tau}_{dim}$. For these 3 metrics, the pseudo inverse estimator yields a significant difference compared to the difference-in-means estimator, indicating that the difference-in-means estimator might underestimate the interference effect. In the remaining metrics, the pseudo inverse estimator detects 1 significant total treatment effect at a $95\%$ confidence level and 2 at a $90\%$ confidence level, whereas the difference-in-means estimator identifies 7 significant total treatment effects at a $95\%$ confidence level. Considering the variance of the pseudo inverse estimator is larger than that of the difference-in-means estimator, its statistical power is lower when the interference effect is not negligible. We believe that the results presented in Table \ref{table: real experiment} demonstrate the pseudo inverse estimator's ability to discover network effects and serve as a valuable tool in real-world experimentation.

\section{Conclusion}

The pseudo inverse estimator represents a novel methodology for estimating the total treatment effect in the presence of network interference. This approach is versatile and can be adapted to various experimental designs, while also exhibiting good theoretical properties. When a firm decides to utilize the pseudo inverse estimator in real-world experimentation, two critical steps can significantly impact the reliability of the results. 

Firstly, the firm must identify an interference network where the estimator can be applied, referred to as the surrogate network in this paper. The quality of this surrogate network, characterized by its deviation from the actual interference network, will influence the bias, while the degree of the surrogate network will determine the estimator's variance. Incorporating additional edges into the surrogate network can reduce bias but may lead to increased variance, and vice versa. Accurate estimation heavily depends on the meticulous design of the surrogate network, which relies on the practitioner's domain expertise and experience. It is advisable to employ historical data for validation during the pre-experiment phase.

Secondly, in the post-experiment analysis, the firm requires a reliable variance estimator to ensure trustworthy statistical inference. We propose a new variance estimator that enhances the one in the original paper and investigate its asymptotic properties within our surrogate network framework. Our simulation results indicate that the proposed estimator performs well when the degree of the surrogate network is relatively small. Furthermore, we introduce a novel method for detecting network interference by combining the pseudo inverse estimator with the difference-in-means estimator, thus extending the pseudo inverse estimator to a broader range of application scenarios. Our real-world implementation of the pseudo inverse estimator showcases its potential for practical application.

We acknowledge three limitations of our study. Firstly, due to practical constraints, we focus on the pseudo inverse estimator with parameter $\beta=1$ in this article; however, it remains an open question to derive new results for other choices of $\beta$ under a similar framework. Secondly, our variance estimation may be biased, particularly when there is significant individual heterogeneity and a substantial deviation of the surrogate network from the actual interference network. Further research is needed to develop methods for compensating this bias under network interference. Lastly, constructing the surrogate network under the bias-variance trade-off, as discussed in Section \ref{section: estimator}, remains an unresolved issue. We defer this task to future research to more precisely construct a surrogate network that closely aligns with the actual interference network.


\bibliographystyle{elsarticle-harv}
\bibliography{references}  





\newpage\appendix
\renewcommand{\theequation}{A\arabic{equation}}
\renewcommand{\thelemma}{A.\arabic{lemma}}
\renewcommand{\theproposition}{A.\arabic{proposition}}
\setcounter{equation}{0}

\section{Proofs}

\subsection{Proof of Lemma \ref{lemma: expectation of estimator}}\label{appendix: proof of expectation}
\begin{proof}
    Let $D_i=\left(\frac{z_i}{p}-\frac{1-z_i}{1-p}\right)$, consider
    \begin{align*}
        E(Y_iD_k)=& pE(Y_iD_k|z_k=1)+(1-p)E(Y_iD_k|z_k=0)\\
        =&E(Y_i|z_k=1)-E(Y_i|z_k=0)\\
        =&E(f_i(\vec z_{-k},z_k=1)-f_i(\vec z_{-k},z_k=0))\\
        =&E(\psi_i^k(\vec z_{-k}))
    \end{align*}
Therefore
\[E(\hat{\tau}(\mathcal{G}))=\frac{1}{n}\sum_{i=1}^n\sum_{k\in\mathcal{M}_i}E(Y_iD_k)=\frac{1}{n}\sum_{i=1}^n\sum_{k\in\mathcal{M}_i}E(\psi_i^k(\vec z_{-\{k\}}))\]
\end{proof}

\subsection{Proof of Lemma \ref{lemma: exogenous bias}}\label{appendix: proof of lemma exogenous bias}
\begin{proof} Recalling the definition of TTE, we have
    \[\text{TTE}=\frac{1}{n}\sum_{i=1}^n(Y_i(\vec 1)-Y_i(\vec 0))=\frac{1}{n}\sum_{i=1}^n\sum_{j\in \mathcal{N}_i} w_{ij}.\]
Under the required assumptions, the result follows from
\begin{align*}
    E(\hat\tau_{\mathcal{G}})=&\frac{1}{n}\sum_{i=1}^n\sum_{k\in\mathcal{M}_i}E(\psi_i^k(\vec z_{-\{k\}}))\\
=&\frac{1}{n} \sum_{i=1}^n\sum_{k\in \mathcal{N}_i} w_{ik}\sum_{j\in \mathcal{M}_i}\mathbbm{1}\{k=j\}\\
 =&\frac{1}{n} \sum_{i=1}^n\sum_{k\in \mathcal{N}_i} w_{ik}\mathbbm{1}\{k\in\mathcal{M}_i\}\\
  =&\text{TTE}-\frac{1}{n} \sum_{i=1}^n\sum_{k=1}^n w_{ik}A_{ik}(1-G_{ik})\\
  \ge & (1-\delta)\text{TTE}
\end{align*}
The second equality follows from $E(\frac{z_i}{p}-\frac{1-z_i}{1-p})=0$, and the third follows from the uniform Bernoulli treatment assignment. The inequality follows from Assumption \ref{ass: gap between A, G}.
\end{proof}

\subsection{Proof of Lemma \ref{lemma: endogenous bias}}\label{appendix: proof of endogenous bias}
\begin{proof}
The result follows from TTE=$ \sum_{\beta=0}^{d_{\mathcal{A}}}\bar a_{\beta}$ and 
    \begin{align*}
        E(\hat\tau_{\mathcal{G}})=&\frac{1}{n}E\left( \sum_{i=1}^n\sum_{S\subseteq \mathcal{N}_i}a_{i,S}\prod_{k\in S}z_k\sum_{j\in \mathcal{N}_i}\left(\frac{z_j}{p}-\frac{1-z_j}{1-p}\right)\right)\\
        =&\frac{1}{n} \sum_{i=1}^n\sum_{S\subseteq \mathcal{N}_i}a_{i,S}\sum_{j\in \mathcal{N}_i}E\left(\prod_{k\in S}z_k\left(\frac{z_j}{p}-\frac{1-z_j}{1-p}\right)\right)\\
        =&\frac{1}{n} \sum_{i=1}^n\sum_{S\subseteq \mathcal{N}_i}a_{i,S}\sum_{j\in \mathcal{N}_i}\mathbbm{1}\{j\in S\}p^{|S|-1}\\
=&\frac{1}{n} \sum_{\beta=0}^{d_{\mathcal{A}}}\sum_{i=1}^n\sum_{S\subseteq \mathcal{N}_i: |S|=\beta}a_{i,S}\beta p^{\beta-1}\\
=& \sum_{\beta=0}^{d_{\mathcal{A}}}\beta p^{\beta-1}\bar a_{\beta}\\
    \end{align*}
\end{proof}

\subsection{Proof of Theorem \ref{theorem: variance upper bound}}\label{appendix: proof of variance upperbound}

\begin{proof}
For the brevity of notation, we use $\vec z_{-S}$ to represents the vector excluding entries in the set $S$. Recalling $D_i=\left(\frac{z_i}{p}-\frac{1-z_i}{1-p}\right)$ and $C_0$ be a sufficiently large universal constant that does not depend on $\mathcal{A}$ and $\mathcal{G}$. We use $\mathbbm{1}\{\cdot\}$ to denote a indicator function.

\begin{align*}
    \Var(\hat\tau(\mathcal{G}))=\Var\left(\frac{1}{n}\sum_{i=1}^n Y_i\sum_{j\in \mathcal{M}_i}D_j\right)=\frac{1}{n^2}\sum_{i=1}^n\sum_{j=1}^n \sum_{k\in \mathcal{M}_i}\sum_{l\in \mathcal{M}_j}\Cov(Y_iD_k,Y_jD_l)
\end{align*}

According to Proposition \ref{prop: Cov(Y_iD_k,Y_jD_k) bound}, $\Cov(Y_iD_k,Y_jD_k)\le \frac{C_1}{p(1-p)}$ for a fixed constant $C_1$. Hence
\begin{align*}
    \Var(\hat\tau(\mathcal{G}))\le \underbrace{\frac{C_1}{n^2p(1-p)}\sum_{i=1}^n\sum_{j=1}^n |\mathcal{M}_i\cap\mathcal{M}_j| }_{(\romannumeral1)}+\underbrace{\frac{1}{n^2}\sum_{i=1}^n\sum_{j=1}^n \sum_{k\in \mathcal{M}_i}\sum_{l\in \mathcal{M}_j}\Cov(Y_iD_k,Y_jD_l)\mathbbm{1}\{k\ne l\}}_{(\romannumeral2)}
\end{align*}

\paragraph{Bound (\romannumeral1):} Given that the surrogate network is undirected, we have
\begin{equation}\label{eq: M_i intersect M_j bound}
    \begin{split}
            &\sum_{i=1}^n\sum_{j=1}^n |\mathcal{M}_i\cap\mathcal{M}_j| \\
    =& \sum_{i=1}^n\sum_{j=1}^n\sum_{k=1}^n \mathbbm{1} \{k\in\mathcal{M}_i\}\mathbbm{1} \{k\in\mathcal{M}_j\}\\
=& \sum_{k=1}^n\sum_{i=1}^n\mathbbm{1} \{k\in\mathcal{M}_i\} \sum_{j=1}^n \mathbbm{1} \{k\in\mathcal{M}_j\}\\
=& \sum_{k=1}^n\sum_{i=1}^n\mathbbm{1} \{i\in\mathcal{M}_k\} \sum_{j=1}^n \mathbbm{1} \{j\in\mathcal{M}_k\}\\
=&\sum_{k=1}^n |\mathcal{M}_k|^2\\
\le & n d_{\mathcal{G}}^2
    \end{split}
\end{equation}

\paragraph{Bound (\romannumeral2):} To apply Proposition \ref{prop: Cov(Y_iD_k,Y_jD_l) bound}, we consider the following inequalities
\begin{align*}
    \sum_{i=1}^n\sum_{j=1}^n \sum_{k\in \mathcal{M}_i}\sum_{l\in \mathcal{M}_j}w_{jk}w_{il}\mathbbm{1}\{k\ne l\}
    \le C_0\sum_{i=1}^n\sum_{j=1}^n\sum_{l\in \mathcal{M}_j}w_{il}
    =C_0\sum_{i=1}^n\sum_{l=1}^nw_{il}|\mathcal{M}_l|
    \le nC_0'd_{\mathcal{G}}
\end{align*}
the first inequality is due to $\sum_{k\in \mathcal{M}_i}w_{jk}\mathbbm{1}\{k\ne l\}\le C_0$ $\forall l$. Similarly,
\begin{align*}
        \sum_{i=1}^n\sum_{j=1}^n \sum_{k\in \mathcal{M}_i}\sum_{l\in \mathcal{M}_j}w_{ik}w_{jk}\mathbbm{1}\{k\ne l\}
    \le d_{\mathcal{G}}\sum_{i=1}^n\sum_{j=1}^n\sum_{k\in \mathcal{M}_i}w_{ik}w_{jk}
    =d_{\mathcal{G}}\sum_{i=1}^n\sum_{k\in \mathcal{M}_i}w_{ik}\sum_{j=1}^nw_{jk}\le nC_0d_{\mathcal{G}}\\
            \sum_{i=1}^n\sum_{j=1}^n \sum_{k\in \mathcal{M}_i}\sum_{l\in \mathcal{M}_j}w_{ik}w_{jk}\mathbbm{1}\{k\ne l\}
    \le \sum_{i=1}^n\sum_{j=1}^n\sum_{k\in \mathcal{M}_i}w_{ik}w_{jk}|\mathcal{M}_j|
    =\sum_{j=1}^n|\mathcal{M}_j|\sum_{i=1}^n\sum_{k\in \mathcal{M}_i}w_{ik}w_{jk}\\
    = \sum_{j=1}^n|\mathcal{M}_j|\sum_{k=1}^n w_{jk}\sum_{i\in \mathcal{M}_k}w_{ik}\le B\bar Y\sum_{j=1}^n|\mathcal{M}_j|
\end{align*}
Next,
\begin{align*}
        &\sum_{i=1}^n\sum_{j=1}^n \sum_{k\in \mathcal{M}_i}\sum_{l\in \mathcal{M}_j}(w_{il}w_{ik}+w_{jl}w_{jk})\mathbbm{1}\{k\ne l\}\\
    \le &2\sum_{i=1}^n\sum_{k\in \mathcal{M}_i}w_{ik}\sum_{j=1}^n\sum_{l\in \mathcal{M}_j}w_{il}\\
    =&2\sum_{i=1}^n\sum_{k\in \mathcal{M}_i}w_{ik}\sum_{l=1}^n\sum_{j\in \mathcal{M}_l}w_{il}\\
    \le & 2\sum_{i=1}^n\sum_{k\in \mathcal{M}_i}w_{ik}\sum_{l=1}^nw_{il}|\mathcal{M}_l|\\
    \le & nC_0d_{\mathcal{G}}
\end{align*}

Finally, 
\begin{align*}
        &\sum_{i=1}^n\sum_{j=1}^n \sum_{k\in \mathcal{M}_i}\sum_{l\in \mathcal{M}_j}\sum_{k'\notin \{k,l\}}w_{ik'}w_{jk'}\mathbbm{1}\{k\ne l\}\\
    \le &d_{\mathcal{G}}^2\sum_{i=1}^n\sum_{j=1}^n \sum_{k'=1}^n w_{ik'}w_{jk'}\\
    =&d_{\mathcal{G}}^2\sum_{i=1}^n \sum_{k'=1}^n w_{ik'}\sum_{j=1}^nw_{jk'}\\
    \le & nC_0d_{\mathcal{G}}^2
\end{align*}
Combining the above results, we arrive at the variance upper bound.
\end{proof}

\begin{proposition}
   There exist a fixed constant $C_0$ such that
    \[|\Cov(Y_iD_k,Y_jD_l)|\le C_0(w_{jk}w_{il}+w_{il}w_{ik}+w_{jl}w_{jk}+\sum_{k'=1}^n w_{ik'}w_{jk'}).\]
    \label{prop: Cov(Y_iD_k,Y_jD_l) bound}
\end{proposition}
\begin{proof}
We rely on the following inequality
\begin{align*}
    &|\Cov(Y_iD_k,Y_jD_l)|\\
    =&|E(Y_iD_kY_jD_l)-E(Y_iD_k)E(Y_jD_l)|\\
    \le & |E(Y_iD_kY_jD_l)-E(Y_iD_k|z_l=0)E(Y_jD_l|z_k=0)|+|E(Y_iD_k|z_l=0)E(Y_jD_l|z_k=0)-E(Y_iD_k)E(Y_jD_l)|
\end{align*}

The proof takes two steps to bound each term in the right-hand side of the above inequality. The result follows from combining two bounds together. For notation brevity, we omit the $\vec z_{-\{k,l\}}$ parameter in both $f_i$ and $\psi_i^j$. In other words, we define $f_i(z_k,z_l)=f_i(\vec z_{-\{k,l\}}, z_k,z_l)$ and $\psi_i^j(z_k,z_l)=\psi_i^j(\vec z_{-\{k,l\}}, z_k,z_l)$ for all $i$ and $j$.

\paragraph{Step 1.} Bound $|E(Y_iD_kY_jD_l)-E(Y_iD_k|z_l=0)E(Y_jD_l|z_k=0)|$. 
\begin{equation}
\begin{split}
&E[Y_iD_kY_jD_l]\\
=& E[Y_iY_j|z_k=1,z_l=1]-E[Y_iY_j|z_k=0,z_l=1]\\
&-E[Y_iY_j|z_k=1,z_l=0]+E[Y_iY_j|z_k=0,z_l=0]\\
=&E[f_i( z_k=1,z_l=1)f_j( z_k=1,z_l=1)]\\
&-E[f_i( z_k=0,z_l=1)f_j( z_k=0,z_l=1)]\\
&-E[f_i( z_k=1,z_l=0)f_j( z_k=1,z_l=0)]\\
&+E[f_i( z_k=0,z_l=0)f_j( z_k=0,z_l=0)]
\label{eq: E[Y_iD_kY_jD_l] decompose}
\end{split}
\end{equation}
The first equality is due to the law of total expectation and the second equality is due to Assumption \ref{ass: potential outcomes} and the fact that $ \vec z_{-\{k,l\}}$, $z_k$, $z_l$ are independent. The following equation for real values $a$, $b$, $c$ and $d$ will be used in the subsequent analysis.
\begin{align}
    ab-cd=(a-c)(b-d)+(b-d)c+(a-c)d \label{eq: 4 value diff}
\end{align}

Recalling the definition of $\psi_i^k$ in Assumption \ref{ass: potential outcomes}, use the above equation,
\begin{equation}\label{eq: f_if_j-f_if_j inequality}
\begin{split}
&f_i( z_k=1,z_l=1)f_j( z_k=1,z_l=1)-f_i( z_k=0,z_l=1)f_j( z_k=0,z_l=1)\\
=&\psi_i^k(z_l=1)f_j( z_k=0,z_l=1)+\psi_j^k(z_l=1)f_i( z_k=0,z_l=1)+\psi_i^k(z_l=1)\psi_j^k(z_l=1)\\
\le &\psi_i^k(z_l=0)f_j( z_k=0,z_l=1)+
\psi_j^k(z_l=0)f_i( z_k=0,z_l=1)\\
&+\psi_i^k(z_l=1)\psi_j^k(z_l=1)-C_0( w_{ik}w_{il}+w_{jk}w_{jl})
\end{split}
\end{equation}
in which the inequality is due to Assumption \ref{ass: bounded outcomes} and \ref{ass: potential outcomes}.

Similarly,
\begin{equation}
\begin{split}
    &f_i( z_k=1,z_l=0)f_j( z_k=1,z_l=0)-f_i( z_k=0,z_l=0)f_j( z_k=0,z_l=0)\\
=&\psi_i^k(z_l=0)f_j( z_k=0,z_l=0)
+\psi_j^k(z_l=0)f_i( z_k=0,z_l=0)+\psi_i^k(z_l=0)\psi_j^k(z_l=0)
\label{eq: f_if_j-f_if_j equality}
\end{split}
\end{equation}

Combine (\ref{eq: f_if_j-f_if_j inequality}) and (\ref{eq: f_if_j-f_if_j equality}) together, we get
\begin{align*}
&f_i( z_k=1,z_l=1)f_j( z_k=1,z_l=1)-f_i( z_k=0,z_l=1)f_j( z_k=0,z_l=1)\\
&-f_i( z_k=1,z_l=0)f_j( z_k=1,z_l=0)
+f_i( z_k=0,z_l=0)f_j( z_k=0,z_l=0)\\
\le &\psi_i^k(z_l=0)\psi_j^l(z_k=0)
+\psi_j^k(z_l=0)\psi_i^l(z_k=0)\\
&+\psi_i^k(z_l=1)\psi_j^k(z_l=1)
-\psi_i^k(z_l=0)\psi_j^k(z_l=0)-C_0( w_{ik}w_{il}+w_{jk}w_{jl})\\
\le &\psi_i^k(z_l=0)\psi_j^l(z_k=0)
+ C_0^2w_{jk}w_{il}+C_0^2w_{ik}w_{jk}+C_0^2w_{ik}w_{jk}+C_0( w_{ik}w_{il}+w_{jk}w_{jl})\\
=&\psi_i^k(z_l=0)\psi_j^l(z_k=0)
+ C_0^2(w_{jk}w_{il}+w_{ik}w_{jk}+w_{ik}w_{il}+w_{jk}w_{jl})
\end{align*}

Substitute above inequation in to (\ref{eq: E[Y_iD_kY_jD_l] decompose}), we get
\begin{equation}
\begin{split}
E[Y_iD_kY_jD_l]
\le&E[\psi_i^k(z_l=0)\psi_j^l(z_k=0)]+ C_0^2(w_{jk}w_{il}+w_{ik}w_{jk}+w_{ik}w_{il}+w_{jk}w_{jl})
\label{eq: E[Y_iD_kY_jD_l] bound}
\end{split}
\end{equation}

Notice that
\begin{equation}
    \begin{split}
        E(Y_iD_k|z_l=0)=&E[Y_i|z_k=1,z_l=0]-E[Y_i|z_k=0,z_l=0]\\
        =&E(f_i( z_k=1,z_l=0)-f_i( z_k=0,z_l=0))\\
        =& E(\psi_i^k(z_l=0))
    \end{split}
\end{equation}
Analogously, $E(Y_jD_l|z_k=0)=E(\psi_j^l(z_k=0))$. Then
\begin{align*}
    &E(Y_iD_kY_jD_l)-E(Y_iD_k|z_l=0)E(Y_jD_l|z_k=0)\\
    \le&\Cov(\psi_i^k(z_l=0),\psi_j^l(z_k=0))+C_0^2(w_{jk}w_{il}+w_{ik}w_{jk}+w_{ik}w_{il}+w_{jk}w_{jl})
\end{align*}

We will use Lemma \ref{lemma: associated r.v.} to bound $\Cov(\psi_i^k(z_l=0),\psi_j^l(z_k=0))$. 
Since each coordinate in $\vec z_{-\{k,l\}}$ is independent, $\vec z_{-\{k,l\}}$ is an associated random vector. Also, let $\lambda_i^{kl}=\sum_{j\in\mathcal{N}_i\backslash \{k,l\}}w_{ij}z_j$. Since $|\psi_i^k(z_l=0,z_j=1)-\psi_i^k(z_l=0,z_j=0)|\le C_0w_{ik}w_{ij}$ for all $j\in\mathcal{N}_i\backslash \{k,l\}$, we have $C_0w_{ik}\lambda_i^{kl}\pm\psi_i^k(z_l=0)$ non-decreasing with respect to each argument of $\vec z_{-\{k,l\}}$ (i.e. $\psi_i^k(z_l=0)\ll C_0w_{ik}\lambda_i^{kl}$). Analogously, $\psi_j^l(z_k=0)\ll C_0w_{jl}\lambda_j^{kl}$. Then
\begin{align*}
    \Cov(\psi_i^k(z_l=0),\psi_j^l(z_k=0))\le C_0^2w_{ik}w_{jl}\Cov(\lambda_i^{kl},\lambda_j^{kl})\le C_0w_{ik}w_{jl}\sum_{k'\in \mathcal{N}_i\cap\mathcal{N}_j\backslash \{k,l\}}w_{ik'}w_{jk'}
\end{align*}

\paragraph{Step 2.} 
By the law of total expectation,
\begin{equation}
    \begin{split}
        E(Y_iD_k)=&pE(Y_iD_k|z_l=1)+(1-p)E(Y_iD_k|z_l=0)  \\
=&p\left(E(Y_i|z_k=1,z_l=1)-E(Y_i|z_k=0,z_l=1)\right)\\
&+(1-p)\left(E(Y_i|z_k=1,z_l=0)-E(Y_i|z_k=0,z_l=0)\right)\\
=& pE(\psi_i^k(z_l=1))+(1-p)E(\psi_i^k(z_l=0))\\
E(Y_iD_k|z_l=0)=&E(Y_i|z_k=1,z_l=0)-E(Y_i|z_k=0,z_l=0)=E(\psi_i^k(z_l=0))
    \end{split}
    \label{eq: E(Y_iD_k) and E(Y_iD_k|z_l=0) formula}
\end{equation}

Thus,
\begin{equation}\label{eq: integral diff of psi}
    \begin{split}
        &|E(Y_iD_k)-E(Y_iD_k|z_l=0)|=\left|p\left(\psi_i^k(z_l=1)-\psi_i^k(z_l=0) \right)\right|\le pC_0w_{il}w_{ik}
    \end{split}
\end{equation}

Therefore,
\begin{align*}
&|E(Y_iD_k)E(Y_jD_l)-E(Y_iD_k|z_l=0)E(Y_jD_l|z_k=0)|\\
\le & |E(Y_iD_k)-E(Y_iD_k|z_l=0)||E(Y_jD_l)-E(Y_jD_l|z_k=0)|\\
&+ |E(Y_iD_k|z_l=0)||E(Y_jD_l)-E(Y_jD_l|z_k=0)|\\
&+|E(Y_jD_l|z_k=0)||E(Y_iD_k)-E(Y_iD_k|z_l=0)|\\
\le & p^2C_0^2w_{il}w_{ik} w_{jl}w_{jk} +pL\left(\bar Y p^{-1} w_{jl}w_{jk}+\bar Y p^{-1} w_{il}w_{ik} \right)\\
\le& C_0(w_{il}w_{ik}+w_{jl}w_{jk})
\end{align*}

\end{proof}

\begin{proposition}
    $\Cov(Y_iD_k,Y_jD_k)\le \frac{C_1}{p(1-p)}$, for all $i$, $j$, $k$ and $l$, where $C_1$ is a fixed constant.
    \label{prop: Cov(Y_iD_k,Y_jD_k) bound}
\end{proposition}
\begin{proof}
    \begin{align*}
        &\Cov(Y_iD_k,Y_jD_k)\\
        =&p^{-1}E(Y_iY_j|z_k=1)+(1-p)^{-1}E(Y_iY_j|z_k=0)+E(Y_iD_k)E(Y_j D_k)\\
        \le &\frac{\bar Y^2}{p(1-p)}+E(Y_iD_k)E(Y_j D_k)\\
        \le & \frac{\bar Y^2}{p(1-p)}+C_0
    \end{align*}
    where the second inequality is due to (\ref{eq: E(Y_iD_k) and E(Y_iD_k|z_l=0) formula}), and $C_0$ is a fixed constant.
\end{proof}

\subsection{Proof of Theorem \ref{theorem: variance lower bound}}\label{appendix: proof of var lower bound}
\begin{proof}
Recalling $D_i=\left(\frac{z_i}{p}-\frac{1-z_i}{1-p}\right)$, we have
\begin{align*}
    \Var(\hat\tau(\mathcal{G}))=&\Var\left(\frac{C_0}{n}\sum_{i=1}^n \sum_{j\in \mathcal{M}_i}D_j\right)\\
    =&\frac{C_0^2}{n^2}\sum_{i=1}^n\sum_{j=1}^n \sum_{k\in \mathcal{M}_i}\sum_{l\in \mathcal{M}_j}\Cov(D_k,D_l)\\
    =&\frac{C_0^2}{n^2}\sum_{i=1}^n\sum_{j=1}^n \sum_{k\in \mathcal{M}_i\cap\mathcal{M}_j}\Var(D_k)\\
    =&\frac{C_0^2}{n^2p(1-p)}\sum_{i=1}^n\sum_{j=1}^n |\mathcal{M}_i\cap\mathcal{M}_j|\\
    =&\frac{C_0^2}{np(1-p)}d_{\mathcal{G}}^2.
\end{align*}
The final equation follows from (\ref{eq: M_i intersect M_j bound}) in Appendix \ref{appendix: proof of variance upperbound} and the assumption that $|\mathcal{M}_i|=d_{\mathcal{G}}$ $\forall i$.
\end{proof}

\subsection{Proof of Theorem \ref{theorem: var estimator bias}}\label{appendix: proof of var estimator endogenous bias}
\begin{proof}\textbf{Step 1. }
    By the definition of $ T_i$ and $I_{ij}$, we have $ T_i\le \frac{\bar Y d_{\mathcal{G}}}{p(1-p)}$ and $\sum_{j=1}^n I_{ij}\le d_{\mathcal{G}}^2$. Recalling Theorem \ref{theorem: variance upper bound}, we have 
    \begin{equation}
        \hat\tau(\mathcal{G})-E(\hat\tau(\mathcal{G}))=\frac{1}{n}\sum_{i=1}^n\left( T_i-E( T_i)\right)=O_p\left(\frac{d_{\mathcal{G}}}{\sqrt{np(1-p)}}\right)
    \end{equation}
    Lemma \ref{lemma: appendix var decompose} tells
        \begin{equation}
        \frac{1}{n}\sum_{i=1}^n\left(\tilde T_i-E(\tilde T_i)\right)=O_p\left(\frac{d_{\mathcal{G}}}{\sqrt{np(1-p)}}\right)
    \end{equation}
    Then we can write
    \begin{align*}
        n\hat\sigma_{\mathcal{G}}^2\backslash d_{\mathcal{G}}^2=& \frac{1}{nd_{\mathcal{G}}^2}\sum_{i=1}^n\sum_{j=1}^n [T_i-\hat\tau(\mathcal{G})][T_j-\hat\tau(\mathcal{G})]I_{ij}\\
        =& \frac{1}{nd_{\mathcal{G}}^2}\sum_{i=1}^n\sum_{j=1}^n [ T_i-E(\hat\tau(\mathcal{G}))][ T_j-E(\hat\tau(\mathcal{G}))]I_{ij}\\
        &+\frac{1}{nd_{\mathcal{G}}^2}[E(\hat\tau(\mathcal{G}))-\hat\tau(\mathcal{G})]\sum_{i=1}^n\sum_{j=1}^n [ T_i+ T_j-2E(\hat\tau(\mathcal{G}))]I_{ij}\\
        &+\frac{1}{nd_{\mathcal{G}}^2}[E(\hat\tau(\mathcal{G}))-\hat\tau(\mathcal{G})]^2\sum_{i=1}^n\sum_{j=1}^n I_{ij}\\
        =& \frac{1}{nd_{\mathcal{G}}^2}\sum_{i=1}^n\sum_{j=1}^n [ T_i-E(\hat\tau(\mathcal{G}))][ T_j-E(\hat\tau(\mathcal{G}))]I_{ij}+O_p\left(\frac{d_{\mathcal{G}}^2}{n^{0.5}p^{1.5}(1-p)^{1.5}}\right)
    \end{align*}
\textbf{Step 2. } We next bound the first term in the right hand side of above equation.
        \begin{align*}
        &\frac{1}{nd_{\mathcal{G}}^2}\sum_{i=1}^n\sum_{j=1}^n [ T_i-E(\hat\tau(\mathcal{G}))][ T_j-E(\hat\tau(\mathcal{G}))]I_{ij}\\
        =& \frac{1}{nd_{\mathcal{G}}^2}\sum_{i=1}^n\sum_{j=1}^n [ T_i-E( T_i)][T_j-E( T_j)]I_{ij}\\
        &+\frac{2}{nd_{\mathcal{G}}^2}\sum_{i=1}^n\sum_{j=1}^n [ T_i-E( T_i)][E( T_j)-E(\hat\tau(\mathcal{G}))]I_{ij}\\
        &+\frac{1}{nd_{\mathcal{G}}^2}\sum_{i=1}^n\sum_{j=1}^n [E( T_i)-E(\hat\tau(\mathcal{G}))][E( T_j)-E(\hat\tau(\mathcal{G}))]I_{ij}\tag{$\mathcal{R}_{\mathcal{G}}$}
    \end{align*}
Let $\omega_i=\sum_{j=1}^n[E( T_j)-E(\hat\tau(\mathcal{G}))]I_{ij}$, then by (), $|\omega_i|\le C_0 d_{\mathcal{G}}^2$. Since
\begin{equation}\label{eq: proof of lemma 1 step 2}
    \begin{split}
            &E\left(\left|\frac{1}{n}\sum_{i=1}^n\sum_{j=1}^n [ T_i-E( T_i)][E( T_j)-E(\hat\tau(\mathcal{G}))]I_{ij} \right|\right)\\
    \le & E\left(\left|\frac{1}{n^2}\sum_{i=1}^n [ T_i-E( T_i)]\omega_i \right|^2\right)^{0.5}\\
    =& \left(\frac{1}{n^2}\sum_{i=1}^n\sum_{j=1}^n \omega_i\omega_j\Cov( T_i, T_j)\right)^{0.5}\\
    =& O\left(\frac{d_{\mathcal{G}}^3}{\sqrt{np(1-p)}} \right)
    \end{split}
\end{equation}
where the last equality follows from Appendix \ref{appendix: proof of variance upperbound}. This implies 
\[\frac{2}{nd_{\mathcal{G}}^2}\sum_{i=1}^n\sum_{j=1}^n [ T_i-E( T_i)][E( T_j)-E(\hat\tau(\mathcal{G}))]I_{ij}=O_p\left(\frac{d_{\mathcal{G}}}{\sqrt{np(1-p)}}\right)\]

\textbf{Step 3. } We next bound the following difference
\begin{align*}
       & \frac{1}{nd_{\mathcal{G}}^2}\sum_{i=1}^n\sum_{j=1}^n [ T_i-E( T_i)][ T_j-E( T_j)]I_{ij}-\frac{1}{nd_{\mathcal{G}}^2}\sum_{i=1}^n\sum_{j=1}^n\Cov(\tilde T_i,\tilde T_j)I_{ij}\\
        =& \frac{1}{nd_{\mathcal{G}}^2}\sum_{i=1}^n\sum_{j=1}^n [T_i T_j-E(\tilde T_i\tilde T_j)]I_{ij}+\frac{2}{nd_{\mathcal{G}}^2}\sum_{i=1}^n\sum_{j=1}^n [E(\tilde T_i)E(\tilde T_j)- T_iE( T_j)]I_{ij}\\
        =& \underbrace{\frac{1}{nd_{\mathcal{G}}^2}\sum_{i=1}^n\sum_{j=1}^n [ T_i T_j-E(\tilde T_i\tilde T_j)]I_{ij}}_{(\romannumeral1)}+\underbrace{\frac{2}{nd_{\mathcal{G}}^2}\sum_{i=1}^n [E( T_i)- T_i]\sum_{j=1}^n E( T_j)I_{ij}}_{(\romannumeral2)}
\end{align*}
Analogous to (\ref{eq: proof of lemma 1 step 2}), we have
\[(\romannumeral2)=O\left(\frac{d_{\mathcal{G}}}{\sqrt{np(1-p)}} \right)\]

\begin{align*}
    (\romannumeral1)=&\underbrace{\frac{1}{nd_{\mathcal{G}}^2}\sum_{i=1}^n\sum_{j=1}^n [\tilde T_i\tilde T_j-E(\tilde T_i\tilde T_j)]I_{ij}}_{(\romannumeral3)}
    +\underbrace{\frac{1}{nd_{\mathcal{G}}^2}\sum_{i=1}^n\sum_{j=1}^n [ T_i-\tilde T_i][  T_j-\tilde T_j]I_{ij}}_{(\romannumeral4)}
    +\underbrace{\frac{2}{nd_{\mathcal{G}}^2}\sum_{i=1}^n[ T_i-\tilde T_i]\sum_{j=1}^n \tilde T_jI_{ij}}_{(\romannumeral5)}
\end{align*}

The term (\romannumeral4) can be bounded in probability by
\begin{align*}
    E[|(\romannumeral4)|]\le &\frac{1}{nd_{\mathcal{G}}^2}\sum_{i=1}^n\sum_{j=1}^n | \Cov(T_i-\tilde T_i, T_j-\tilde T_j)I_{ij}|\\
    \le & \frac{1}{nd_{\mathcal{G}}^2}\sum_{i=1}^n\sum_{j=1}^n\sum_{k\in\mathcal{M}_i}\sum_{l\in\mathcal{M}_j}|\Cov(T_{ik}-\tilde T_{ik},T_{jl}-\tilde T_{jl})|\\
    =& O\left(\frac{\delta^2}{p(1-p)}\right)
\end{align*}
where the last equality can be obtained using the same procedure in Lemma \ref{lemma: appendix var decompose}.

Let $\tilde \omega_i=\sum_{j=1}^n \tilde T_jI_{ij}$, then $|\tilde \omega_i|\le C_0 d_{\mathcal{G}}^3$ . Similarly,
\begin{align*}
    E[|(\romannumeral5)|]=&\frac{2}{d_{\mathcal{G}}^2}E\left(\frac{1}{n}\sum_{i=1}^n[ T_i-\tilde T_i]\tilde \omega_i\right)\\
    \le & \frac{2}{d_{\mathcal{G}}^2}E\left(\left|\frac{1}{n}\sum_{i=1}^n[ T_i-\tilde T_i]\tilde \omega_i\right|^2\right)^{0.5}\\
    =&\frac{2}{d_{\mathcal{G}}^2}\left(\frac{1}{n^2}\sum_{i=1}^n\sum_{j=1}^n\Cov(T_i-\tilde T_i,T_j-\tilde T_j)\tilde \omega_i\tilde \omega_j\right)^{0.5}\\
    =&O\left(\frac{\delta d_{\mathcal{G}}^2}{\sqrt{np(1-p)}}\right)
\end{align*}

Finally, since $\tilde T_i$ and $\tilde T_j$ are independent if $I_{ij}=0$ for all $i$ and $j$, we have
$\Cov(\tilde T_i\tilde T_jI_{ij},\tilde T_k\tilde T_lI_{kl})=0$ when $(1-I_{ik})(1-I_{jk})(1-I_{il})(1-I_{jl})=1$. Also, $|\Cov(\tilde T_i\tilde T_jI_{ij},\tilde T_k\tilde T_lI_{kl})|\le C_0 \left(\frac{d_{\mathcal{G}}}{p(1-p)}\right)^4$. Thus we have
\begin{align*}
    \Var(\romannumeral3)=&\frac{1}{n^2d_{\mathcal{G}}^4}\sum_{i=1}^n\sum_{j=1}^n\sum_{k=1}^n\sum_{l=1}^n\Cov(\tilde T_i\tilde T_jI_{ij},\tilde T_k\tilde T_lI_{kl})\\
    \le &\frac{1}{n^2p^4(1-p)^4} \sum_{i=1}^n\sum_{j=1}^n\sum_{k=1}^n\sum_{l=1}^n I_{ij}I_{kl}(I_{ik}+I_{jk}+I_{il}+I_{jl})\\
    = &\frac{4}{n^2p^4(1-p)^4} \sum_{i=1}^n\sum_{j=1}^n\sum_{k=1}^n\sum_{l=1}^n I_{ij}I_{kl}I_{ik}\\
    = &\frac{4}{n^2p^4(1-p)^4} \sum_{i=1}^n\sum_{j=1}^nI_{ij}\sum_{k=1}^nI_{ik}\sum_{l=1}^n I_{kl}\\
    = & O\left(\frac{d_{\mathcal{G}}^6}{np^4(1-p)^4}\right)
\end{align*}
which means
\begin{align*}
    (\romannumeral3)=O_p\left(\frac{d_{\mathcal{G}}^3}{\sqrt{n}p^2(1-p)^2}\right)
\end{align*}
\textbf{Step 4. } Finally, we use Lemma \ref{lemma: appendix var decompose} to bound
\[\left|\frac{1}{n^2}\sum_{i=1}^n\sum_{j=1}^n\Cov(\tilde T_i,\tilde T_j)I_{ij}-\Var(\hat\tau(\mathcal{G}))\right|=O\left(\frac{\delta d_{\mathcal{G}}^2}{np(1-p)}\right)\]
Combine each bound in Step 1 to 3, the result follows.
\end{proof}

\begin{lemma}\label{lemma: appendix var decompose}
 \[\frac{1}{n^2}\sum_{i=1}^n\sum_{j=1}^n\sum_{k\in\mathcal{M}_i}\sum_{l\in\mathcal{M}_j}\Cov(\tilde T_{ik},\tilde T_{jl})=\Var(\hat\tau(\mathcal{G}))+O\left(\frac{\delta d_{\mathcal{G}}^2}{np(1-p)}\right)\]
\end{lemma}
\begin{proof} Let $\tilde T_{ik}=E(f_i(\vec z)D_{k}|\vec z_{\mathcal{M}_i})=E(f_i(\vec z)|\vec z_{\mathcal{M}_i})D_{k}$. 

 \[\Var(\hat\tau(\mathcal{G}))=\frac{1}{n^2}\sum_{i=1}^n\sum_{j=1}^n\sum_{k\in\mathcal{M}_i}\sum_{l\in\mathcal{M}_j}[\Cov(\tilde T_{ik},\tilde T_{jl})+\Cov(\tilde T_{ik},T_{jl}-\tilde T_{jl})+\Cov(T_{ik}-\tilde T_{il},T_{jl})]\]
 
Let $\tilde f_i(\vec z)=f_i(\vec z)-E(f_i(\vec z)|\vec z_{\mathcal{M}_i})$ and $\tilde\psi_i^k(\vec z_{-\{k\}})=\tilde f_i(\vec z_{-\{k\}},z_k=1)-\tilde f_i(\vec z_{-\{k\}},z_k=0)=\psi_i^k(\vec z_{-\{k\}})-E(\psi_i^k(\vec z_{-\{k\}})|\vec z_{\mathcal{M}_i\backslash \{k\}})$. By Assumption \ref{ass: bounded outcomes} and \ref{ass: potential outcomes},
\begin{align*}
&f_i(\vec z_{\mathcal{M}_i},\vec z_{-\mathcal{M}_i})- f_i(\vec z_{\mathcal{M}_i},\vec h_{-\mathcal{M}_i})\\
 \le& C_0\sum_{j\in \mathcal{M}_i} w_{ij}=C_0\sum_{j=1}^n w_{ij}\max\{A_{ij}-G_{ij},0\}\le \delta C_0\\
 &\forall i,\;\vec z_{\mathcal{M}_i},\vec h_{-\mathcal{M}_i}\in \{0,1\}^{n-|\mathcal{M}_i|}
\end{align*}
Analogously, 
\begin{align*}
 &\psi_i^k(\vec z_{\mathcal{M}_i\backslash \{k\}},\vec z_{-\mathcal{M}_i\cup \{k\} })- \psi_i^k(\vec z_{\mathcal{M}_i\backslash \{k\}},\vec h_{-\mathcal{M}_i\cup \{k\}})\\
 \le& C_0w_{ik}\sum_{j\in \mathcal{M}_i} w_{ij}=C_0w_{ik}\sum_{j=1}^n w_{ij}\max\{A_{ij}-G_{ij},0\}\le \delta C_0w_{ik}, \\
 &\forall i,k,\;\vec z_{\mathcal{M}_i\backslash  \{k\}},\vec h_{-\mathcal{M}_i\cup \{k\}}\in \{0,1\}^{n-|\mathcal{M}_i\cup\{k\}|}
\end{align*}
Finally, by Assumption \ref{ass: Lipchitz second derivative Y_i},
\begin{align*}
 &\phi_i^{kl}(\vec z_{\mathcal{M}_i\backslash \{k,l\}},\vec z_{-\mathcal{M}_i\cup \{k,l\} })- \phi_i^{kl}(\vec z_{\mathcal{M}_i\backslash \{k,l\}},\vec h_{-\mathcal{M}_i\cup \{k,l\}})\\
 \le& C_0w_{ik}w_{il}\sum_{j\in \mathcal{M}_i} w_{ij}=C_0w_{ik}w_{il}\sum_{j=1}^n w_{ij}\max\{A_{ij}-G_{ij},0\}\le \delta C_0w_{ik}w_{il}, \\
 &\forall i,k\ne l,\;\vec z_{\mathcal{M}_i\backslash \{k,l\}},\vec h_{-\mathcal{M}_i\cup \{k,l\}}\in \{0,1\}^{n-|\mathcal{M}_i\cup\{k,l\}|}
\end{align*}

Then, for all $i$ and $k\ne l$ we have
\begin{align*}
|\tilde f_i(\vec z)|=&|f_i(\vec z)-E(f_i(\vec z)|\vec z_{\mathcal{M}_i})|\le \delta C_0\\
 |\tilde\psi_i^k(\vec z_{-\{k\}})|=&|f_i(\vec z_{-\{k\}},z_k=1)-E(f_i(\vec z)|\vec z_{\mathcal{M}_i\backslash \{k\}},z_k=1)-f_i(\vec z_{-\{k\}},z_k=0)+E(f_i(\vec z)|\vec z_{\mathcal{M}_i\backslash \{k\}},z_k=0)| \\
 = &|\psi_i^k(\vec z_{-\{k\}})-E(\psi_i^k(\vec z_{-\{k\}})|\vec z_{\mathcal{M}_i\backslash \{k\}})|\\
 \le & \delta C_0w_{ik} \\
  |\tilde\psi_i^k(\vec z_{-\{k,l\}},z_l=1)&-\tilde\psi_i^k(\vec z_{-\{k,l\}},z_l=0)|
  =|\psi_i^k(\vec z_{-\{k,l\}},z_l=1)-E(\psi_i^k(\vec z_{-\{k,l\}},z_l=1)|\vec z_{\mathcal{M}_i\backslash \{k,l\}})\\
  &-\psi_i^k(\vec z_{-\{k,l\}},z_l=0)+E(\psi_i^k(\vec z_{-\{k,l\}},z_l=0)|\vec z_{\mathcal{M}_i\backslash \{k,l\}})| \\
 = &|\phi_i^{k,l}(\vec z_{-\{k,l\}})-E(\phi_i^{k,l}(\vec z_{-\{k,l\}})|\vec z_{\mathcal{M}_i\backslash \{k,l\}})|\\
 \le & \delta C_0w_{ik} w_{il}
\end{align*}

Now, consider $\Cov(\tilde T_{ik},T_{jl}-\tilde T_{jl})=\Cov(g_i(\vec z) D_{ik},\tilde f_j(\vec z)D_{jl})$, where $g_i(\vec z)=E(f_i(\vec z)D_{k}|\vec z_{\mathcal{M}_i})$. Obviously, $g_i$ satisfies Assumption \ref{ass: potential outcomes}. We replace the $f_i$, $f_j$, $\psi_j^k$ and $\psi_j^l$ in Proposition \ref{prop: Cov(Y_iD_k,Y_jD_l) bound} by $g_i$, $\tilde f_j$, $\tilde\psi_j^k$ and $\tilde\psi_j^l$, respectively. Applying the bounds derived above to $|\tilde f_j(\vec z)|$, $|\tilde\psi_j^k(\vec z_{-\{k\}})|$ and $|\tilde\psi_j^l(\vec z_{-\{l\}})|$ and following the procedure in Proposition \ref{prop: Cov(Y_iD_k,Y_jD_l) bound}, we can get exactly the same bound for $|\Cov(\tilde T_{ik},T_{jl}-\tilde T_{jl})|$ except the constant $C_0$ shrinks to $\delta C_0$. Similarly, the bound in Propostion \ref{prop: Cov(Y_iD_k,Y_jD_k) bound} also shrinks by $\delta$. Then following the steps in Appendix \ref{appendix: proof of variance upperbound}, we get
\begin{align*}
\frac{1}{n^2}\sum_{i=1}^n\sum_{j=1}^n\sum_{k\in\mathcal{M}_i}\sum_{l\in\mathcal{M}_j}|\Cov(\tilde T_{ik},T_{jl}-\tilde T_{jl})|=O\left(\frac{\delta d_{\mathcal{G}}^2}{np(1-p)}\right)
\end{align*}
Similar technique can be apply to $|\Cov(T_{ik}-\tilde T_{il},T_{jl})|$ and the result follows.
\end{proof}

\begin{lemma}[\citealp{newman1984-associated-rv}]
\label{lemma: associated r.v.}
For a pair of measurable numeric functions $f$ and $g$ defined on $A\in R^k$, we write $f\ll g$ if both functions $g+f$ and $g-f$ are nondecreasing with respect to each argument. Now let $X$ be any associated random vector with range in $A$. Then
\begin{align*}
    (f_i\ll g_i \text{ for } i=1,2)\Rightarrow  (|\Cov(f_1(X),f_2(X))|\le \Cov(g_1(X),g_2(X)))
\end{align*}
\end{lemma}

\end{document}